\documentclass[journal]{IEEEtran}
\usepackage[mathscr]{eucal}
\usepackage[cmex10]{amsmath}
\usepackage{epsfig,epsf,psfrag}
\usepackage{amssymb,amsthm,amsfonts,latexsym}
\usepackage{graphicx,bm,xcolor,url,overpic}
\usepackage{array}
\usepackage{verbatim}
\usepackage{bm}
\usepackage{algorithmic}
\usepackage{algorithm}
\usepackage{verbatim}
\usepackage{textcomp}
\usepackage{mathrsfs}
\usepackage{epstopdf}
\usepackage{bbm,cite}
\usepackage{hyperref}
\usepackage{tikz}
\usetikzlibrary{positioning}

\ifCLASSOPTIONcompsoc
    \usepackage[caption=false, font=normalsize, labelfont=sf, textfont=sf]{subfig}
\else
\usepackage[caption=false, font=footnotesize]{subfig}
\fi
\usepackage{enumitem}

\newtheorem{theorem}{Theorem} 
\newtheorem{lemma}{Lemma}

\newtheorem{definition}{Definition}

\newtheorem{remark}{Remark}

\usepackage[utf8]{inputenc}

\title{Covert Communication Gains from Adversary's Uncertainty of Phase Angles}
\author{
\IEEEauthorblockN{Sen Qiao}, \IEEEauthorblockN{Daming Cao}, \IEEEauthorblockN{Qiaosheng Zhang}, \IEEEauthorblockN{Yinfei Xu,~\IEEEmembership{Member,~IEEE,}} and \IEEEauthorblockN{Guangjie Liu}
\thanks{This work was supported by the National Key R\&D Program of China (Grants No.\ 2021QY0700), the National Natural Science Foundation of China (Grants No.\ U21B2003,\ 62072250), the Startup Foundation for Introducing Talent of NUIST (Grants No.\ 2023r014) and Zhi Shan Young Scholar Program of Southeast University.}

\thanks{Sen Qiao, Daming Cao and Guangjie Liu are with the School of Electrical and Information Engineering, Nanjing University of Information Science and Technology, Nanjing, 210044, China (e-mail: sensariel@nuist.edu.cn; dmcao@nuist.edu.cn; gjieliu@gmail.com).}
\thanks{Qiaosheng Zhang is with Shanghai Artificial Intelligence Laboratory, Shanghai, 200032, China (e-mail: zhangqiaosheng@pjlab.org.cn).}
\thanks{Yinfei Xu is with the School of Information Science and Engineering, Southeast University, Nanjing 210096, China (e-mail: yinfeixu@seu.edu.cn).}
\thanks{Corresponding author: Daming Cao.}
}
\date{October 2022}

\begin{document}

\maketitle
\allowdisplaybreaks[1]

\begin{abstract}
    This work investigates the phase gain of intelligent reflecting surface (IRS) covert communication over complex-valued additive white Gaussian noise (AWGN) channels. The transmitter Alice intends to transmit covert messages to the legitimate receiver Bob via reflecting the broadcast signals from a radio frequency (RF) source, while rendering the adversary Willie's detector arbitrarily close to ineffective. Our analyses show that, compared to the covert capacity for classical AWGN channels, we can achieve a covertness gain of value 2 by leveraging Willie's uncertainty of phase angles. This covertness gain is achieved when the number of possible phase angle pairs $N=2$. More interestingly, our results show that the covertness gain will not further increase with $N$ as long as $N \ge 2$, even if it approaches infinity.
\end{abstract}

\section{Introduction}
\IEEEPARstart{I}{n} certain complicated antagonistic realms (such as military communications), even a little exposed intention of communication may lead to significant strategic failures. Consequently, the military has developed diverse techniques (e.g. the spread spectrum technique \cite{simon2002spread,Chen2023Covert,Bash2015Hiding,cek2009stable}) to ensure the covertness of communication, i.e., to hide the very presence of communication from watchful adversaries. From the theoretical perspective, the information-theoretic limit of covert communication was first investigated by \cite{bash2012square}, which discovered a $\textit{square root law}$ (SRL) for additive white Gaussian noise (AWGN) channels. This seminal theorem has subsequently been extended to various channel models, including binary symmetric channels\cite{che2013reliable}, discrete memoryless channels\cite{wang2016fundamental,Tahmasbi2019first,bloch2016covert} and multiuser channels\cite{arumugam2019covert,Kumar2019embedding,cho2021treating,Kibloff2019embedding,Tan2019time}, etc.

In the covert communication scenario, the transmitter Alice occasionally wishes to transmit a message to the legitimate receiver Bob over a noisy channel, while simultaneously ensures that the adversary Willie is not able to detect the transmission (if exists).  The SRL states that to ensure both covertness and reliability, only $\mathcal{O}\big(\sqrt{n}\big)$ bits can be transmitted over $n$ channel uses. Note that the transmission rate approaches zero as $n$ grows to infinity.

\textcolor{black}{Prior works have put forth diverse strategies to improve the performance of covert communication, including relaying networks\cite{Hu2018Covert,Forouzesh2020Covert}, multiple interference networks\cite{he2018covert,zheng2019multi}, unmanned aerial vehicle (UAV) networks\cite{Yan2021Optimal,zhou2019joint}, multi-user networks\cite{huang2021jamming}, etc.} In particular, Lu et al.\cite{lu2020intelligent} noticed that the intelligent reflecting surface (IRS) \textcolor{black}{(a.k.a. the reconfigurable intelligent surface (RIS))} has the capability of enhancing the received signal at the receiver side while simultaneously deteriorating the signal at the warden side. In their setting, Alice transmits messages covertly by reflecting her signal to Bob, or reflecting additional noise to Willie via IRS devices. Following their pioneering work, recent works \cite{wang2021energy,liu2022covert,ma2022covert} further show that for IRS covert communication, Willie's uncertainty about noise can be appropriately leveraged to enhance the covert performance. Besides, the optimization of transmission power and reflection beamforming in IRS networks have also been investigated in \cite{wu2021intelligent,lv2022covert,Wu2019Intelligent} and \cite{si2021covert,zhou2021joint,wang2021energy}, respectively. Moreover, the covert communication in UAV mounted IRS (UIRS) communication systems has also been investigated in \cite{Tatar2022Aerial,wang2023covert}.

Different from the methods in the aforementioned works \cite{Hu2018Covert,Forouzesh2020Covert,Yan2021Optimal,he2018covert,zheng2019multi,Yan2019Delay-Intolerant,huang2021jamming,zhou2019joint,lu2020intelligent,wu2021intelligent,lv2022covert,si2021covert,zhou2021joint,Wu2019Intelligent,Tatar2022Aerial,wang2023covert,wang2021energy,liu2022covert,ma2022covert}, utilizing other resources, such as the spectrum and time resource, has also been proven to be effective approaches to improve the performance of covert communication. In \cite{wang2021covert}, Wang et al. investigated the problem of covert communication over Multiple-Input Multiple-Output (MIMO) AWGN channels, where all users are equipped with multiple antennas. Furthermore, the authors in \cite{bash2016covert} considered utilizing the time resource to enhance the performance of covert communication. In their setting, Alice and Bob are allowed to secretly choose one single time slot (out of $T(n)$ slots) to communicate, and \cite{bash2016covert} showed that they can transmit $\mathcal{O}\Big(\min\{\sqrt{n\log{T(n)}},{n}\}\Big)$ bits reliably and covertly when Willie does not have the knowledge of the chosen slot. 
\begin{figure*}[ht]
    \centering
\subfloat[Not deflect symbols]{
    	\begin{tikzpicture}[
		roundnode/.style={ circle, draw=white!100, fill=black!100, very thick, minimum size=1mm},
		squarednode/.style={rectangle, draw=black!100, fill=white!100, thick, minimum size=30mm, dotted},
		pointnode/.style={rectangle, draw=white!1, very thick, minimum size=1mm},
		HHnode/.style={align=center, dotted,rectangle, draw=black!100, fill=white!100, very thick, minimum size=10mm},
]
\node[roundnode]        (zero1)         {}   ;
\node[pointnode]        (zr1)       [right=1cm of zero1]   {}   ;
\node[pointnode]        (zl1)       [left=1cm of zero1]   {}   ;
\node[pointnode]        (za1)       [above=1cm of zero1]   {}   ;
\node[pointnode]        (zb1)       [below=1cm of zero1]   {}   ;

\node[squarednode]        (kuang)         [at =(zero1)]  {} ;
\node[pointnode]        (zero)         [right=3cm of zero1]  {} ;
\node[pointnode]        (zr)       [right=1cm of zero]   {}  ;
\node[pointnode]        (zl)       [left=1cm of zero]   {}   ;
\node[pointnode]        (za)       [above=1cm of zero]   {}   ;
\node[pointnode]        (zb)       [below=1cm of zero]   {}   ;
\node[squarednode]        (kuang2)         [at =(zero)]  {} ;
\node at (zero1)[circle,fill,inner sep=2pt]{};
\node[pointnode]        (zar)       [above right =0.75cm and 0.75cm of zero]   {}  ;
\node[pointnode]        (zbl)       [below left =0.75cm and 0.75cm of zero]   {}   ;
\node  (01) at ([shift={(1em,-0.5em)}]za1) {$y$};
\node  (02) at ([shift={(-0.5em,1em)}]zr1) {$x$};
\node  (05) at ([shift={(0.5em,2.8em)}]zl1) {$\mathcal{H}_0$};
\node  (06) at ([shift={(0.5em,2.8em)}]zl) {$\mathcal{H}_1$};
\node  (03) at ([shift={(1em,-0.5em)}]za) {$y$};
\node  (04) at ([shift={(-0.5em,1em)}]zr) {$x$};

\draw[->,line width=0.4mm] (zl) -- (zr);
\draw[->,line width=0.4mm] (zb) -- (za);
\draw[->,line width=0.4mm] (zl1) -- (zr1);
\draw[->,line width=0.4mm] (zb1) -- (za1);
\draw[latex-latex,line width=0.5mm] (zbl) -- (zar);
\end{tikzpicture} }\qquad
\subfloat[Deflect each symbol]{
	\begin{tikzpicture}[
		roundnode/.style={ circle, draw=white!100, fill=black!100, thick, minimum size=1mm},
		squarednode/.style={rectangle, draw=black!100, fill=white!100, thick, minimum size=30mm, dotted},
		pointnode/.style={rectangle, draw=white!1, very thick, minimum size=1mm},
		HHnode/.style={align=center, dotted,rectangle, draw=black!100, fill=white!100, very thick, minimum size=10mm},
]
\node[pointnode]        (zero)         []  {} ;
\node[pointnode]        (zr)       [right=1cm of zero]   {}  ;
\node[pointnode]        (zl)       [left=1cm of zero]   {}   ;
\node[pointnode]        (za)       [above=1cm of zero]   {}   ;
\node[pointnode]        (zb)       [below=1cm of zero]   {}   ;
\node[squarednode]        (kuang2)         [at= (zero)]  {} ;
\node[pointnode]        (zar)       [above right =0.75cm and 0.75cm of zero]   {}  ;
\node[pointnode]        (zbl)       [below left =0.75cm and 0.75cm of zero]   {}   ;
\node[pointnode]        (zal)       [below right =0.75cm and 0.75cm of zero]   {}  ;
\node[pointnode]        (zbr)       [above left =0.75cm and 0.75cm of zero]   {}   ;
\node  (06) at ([shift={(0.5em,2.8em)}]zl) {$\mathcal{H}_1$};
\node  (03) at ([shift={(1em,-0.5em)}]za) {${y}$};
\node  (04) at ([shift={(-0.5em,1em)}]zr) {$x$};


\draw[->,line width=0.4mm] (zl) -- (zr);
\draw[->,line width=0.4mm] (zb) -- (za);

\draw[latex-latex,line width=0.5mm] (zbl) -- (zar);

\draw[latex-latex,line width=0.5mm] (zbr) -- (zal);
\end{tikzpicture}}\qquad
\subfloat[Deflect $n$ symbols]{
\begin{tikzpicture}[
		roundnode/.style={ circle, draw=white!100, fill=black!100, very thick, minimum size=1mm},
		squarednode/.style={rectangle, draw=black!100, fill=white!100, thick, minimum height=30mm, minimum width=6cm, dotted},
		pointnode/.style={rectangle, draw=white!1, very thick, minimum size=1mm},
		HHnode/.style={align=center, dotted,rectangle, draw=black!100, fill=white!100, very thick, minimum size=10mm},
]
\node[pointnode]        (zero1)         {}   ;
\node[pointnode]        (zr1)       [right=1cm of zero1]   {}   ;
\node[pointnode]        (zl1)       [left=1cm of zero1]   {}   ;
\node[pointnode]        (za1)       [above=1cm of zero1]   {}   ;
\node[pointnode]        (zb1)       [below=1cm of zero1]   {}   ;

\node[pointnode]        (zero)         [right=2.5cm of zero1]  {} ;
\node[pointnode]        (zr)       [right=1cm of zero]   {}  ;
\node[pointnode]        (zl)       [left=1cm of zero]   {}   ;
\node[pointnode]        (za)       [above=1cm of zero]   {}   ;
\node[pointnode]        (zb)       [below=1cm of zero]   {}   ;
\node[squarednode]        (kuang2)         [at= (zl)]  {} ;

\node[pointnode]        (zar)       [above right =0.75cm and 0.75cm of zero1]   {}  ;
\node[pointnode]        (zbl)       [below left =0.75cm and 0.75cm of zero1]   {}   ;

\node[pointnode]        (zbr)       [below right =0.75cm and 0.75cm of zero]   {}  ;
\node[pointnode]        (zal)       [above left =0.75cm and 0.75cm of zero]   {}   ;

\node  (01) at ([shift={(1em,-0.5em)}]za1) {$y$};
\node  (02) at ([shift={(-0.5em,1em)}]zr1) {$x$};

\node  (05) at ([shift={(0.5em,2.8em)}]zl1) {$\mathcal{H}_1$};

\node  (03) at ([shift={(1em,-0.5em)}]za) {$y$};
\node  (04) at ([shift={(-0.5em,1em)}]zr) {$x$};

\node  (06) at ([shift={(0.2em,0em)}]zr1) {or};


\draw[->,line width=0.4mm] (zl) -- (zr);
\draw[->,line width=0.4mm] (zb) -- (za);

\draw[->,line width=0.4mm] (zl1) -- (zr1);
\draw[->,line width=0.4mm] (zb1) -- (za1);

\draw[latex-latex,line width=0.5mm] (zbl) -- (zar);
\draw[latex-latex,line width=0.5mm] (zbr) -- (zal);

\end{tikzpicture}}
    \caption{Hypothesis testing with and without phase}
    \label{fig:my_label}
\end{figure*}
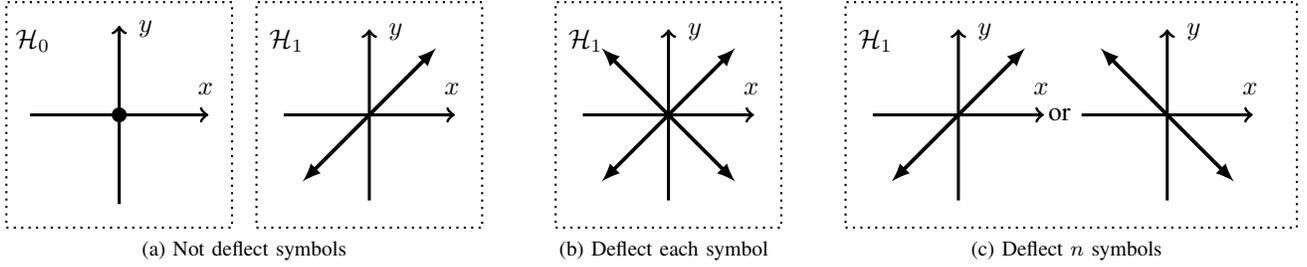

In addition to spectrum and time, \emph{phase} is another resource that can be utilized, however many communication scenarios studied in literature locate on real-valued AWGN channels\cite{bash2012square,bash2013limits,bash2016covert,bloch2016covert}. Further, to the best of our knowledge, no work has realized the benefits of utilizing phase in covert communication. Since the IRS can transmit information by varying the amplitude and/or phase of signals\cite{renzo2020smart}, it is interesting to investigate whether one can transmit more covert information by utilizing the phase resource via IRS. We note that most existing  works on IRS covert communication focus on the optimization of transmission power and reflection beamforming, without taking the effect of phases into consideration. For example, the pioneers in \cite{chen2021on} considered IRS covert communication over complex-valued AWGN channels with Binary Phase Shift Keying (BPSK) codebook, but no phase changes has been employed. Here we show how Alice can leverage Willie's ignorance of the exact phase angle to improve the covertness.

In our scenario, Alice communicates to Bob with an IRS device and a shared secret of sufficient length, and Willie may not know the exact phase angles that Alice and Bob select in advance. Willie observes over a complex-valued Gaussian channel and performs a binary hypothesis test to detect the communication. When Alice transmits with a BPSK codebook, like the setting in \cite{bash2013limits,bash2016covert,chen2021on}, Willie can perform a binary hypothesis test as shown in Fig.~1(a). For ease of notation, we express the phase angle with values between zero and $2\pi$. The whole transmission consists of $n$ symbols. All symbols are sent with two supplementary initial angles; for concreteness we set the angles to be $\pi/4$ and $5\pi/4$. The choice of different initial angles does not change Willie's ability of detection due to the symmetry of Gaussian noise. However, by utilizing IRS, Alice can also reflect each symbol with other phase angles except for the two initial angles, e.g., $3\pi/4$ and $7\pi/4$ (this is equivalent to deflect the initial angles by $\pi/2$), as shown in Fig.~1(b). Thus, Willie does not know whether the initial angles $(\pi/4, 5\pi/4)$ or the deflected angles $(3\pi/4, 7\pi/4)$ are used by Alice. Since this increases the uncertainty of Willie, Alice can achieve an improvement over a naive application of the SRL. Perhaps surprisingly, we show that it is not necessary for Alice to reflect each symbol with a different angle. Instead, Alice can achieve an improvement by confusing Willie whether all the $n$ symbols are deflected by a same phase angle (e.g., $\pi/2$) or not, as shown in Fig.~1(c). We refer the readers to Sec. \ref{Codebook Construction} and Sec. \ref{Hypothesis test} for more details about the two deflection methods.

 Deflect, or not deflect? Just one-bit information can provide a substantial covertness gain. The improvement stems from the fact that Willie does not know the exact phase and has to detect all possible phase angles. In this work, we investigate the effect of phases in covert communication for the first time. We name the covertness improvement achieved from phase resources as {\em phase gain}. Our results highlight the phase gain by comparing covert performance between codebooks with a single phase pair and codebooks with multiple phase pairs. We provide detailed achievability proofs of three different codebooks. In particular, we would like to point out that the calculation of the KL divergence of one particular codebook (with multiple phase pairs) requires non-trivial analytical techniques, since its induced output distribution cannot be transformed to $n$ single-letter distributions by the chain rule. To circumvent this difficulty, we take a non-trivial approach---using the Taylor series expansion $\displaystyle\lim_{x \to 0}\log{(1+x)}=x-\frac{1}{2}x^2+\mathcal{O}(x^3)$ to approximate the KL divergence so that it can be calculated by summing up the approximations of three terms. Moreover, we generalize our results to infinite phase angles and prove that further increasing the number of phase angles does not lead to a larger phase gain.

The rest of the paper is organized as follows. In Section II, we introduce the system model and codebooks for IRS covert communication over complex-valued AWGN channels. Section~III provides the main results of this work. \textcolor{black}{In Section IV, we present the detailed proofs of our results. Section V presents numerical results that validate our theoretical results. Section VI concludes this work and proposes several directions that are fertile avenues for future research.}

\section{Prerequisites}
\subsection{System Model}
We consider a complex-valued discrete-time AWGN channel model, as shown in Fig.~2. The RF source continuously broadcasts $n$ complex-valued random variable $\mathbf{S_{{R}}}=\{S_{R,i}\}^n_{i=1}$. The IRS transmitter (Alice) intends to transmit $n$ complex-valued symbols $\mathbf{c}=\{c_{i}\}^n_{i=1}$ to a receiver (Bob) by utilizing the broadcast signal from the RF source, while the detector (Willie) seeks to detect the existence of the transmission. Then, the $i$-th symbol that Bob and Willie received can be expressed as
\begin{align}
    &S_{{B,i}}={{h_{RB}}}S_{R,i}+{{h_{RA}}}{{h_{AB}}}{S_{R,i}c_{i}}+Z_{B,i},\\
    &S_{W,i}={{h_{RW}}}S_{R,i}+{{h_{RA}}}{{h_{AW}}}{S_{R,i}c_{i}}+Z_{W,i},
\end{align}
where $Z_{B,i}$ and $Z_{W,i}$ are independent and identically distributed (i.i.d.) Gaussian noise with variance $2\sigma^2$, i.e., $Z_{B,i},Z_{W,i}\sim \mathcal{CN}(0,2\sigma^2)$. The channel coefficients from the RF source to Alice, Bob and Willie are denoted by ${{h_{RA}}}$, ${{h_{RB}}}$ and ${{h_{RW}}}$, respectively. The channel coefficients from Alice to Bob and to Willie are denoted by ${{h_{AB}}}$ and ${{h_{AW}}}$, respectively. Without loss of generality, we assume the RF source broadcasts $n$ symbols with zero phase angle, and all channel coefficients are known to everyone.

\begin{figure}[htbp]
	\center
	\begin{tikzpicture}[
		roundnode/.style={ circle, draw=black!100, fill=white!100, very thick, minimum size=10mm},
		squarednode/.style={rectangle, draw=black!100, fill=white!100, very thick, minimum height=10mm, minimum width=20mm},
		RF/.style={rectangle, draw=black!100, fill=green!30, very thick, minimum height=10mm, minimum width=20mm},
		willie/.style={rectangle, draw=black!100, fill=blue!20, very thick, minimum height=10mm, minimum width=20mm},
		AB/.style={rectangle, draw=black!100, fill=red!30, very thick, minimum height=10mm, minimum width=20mm},
		pointnode/.style={rectangle, draw=white!1, very thick, minimum size=1mm},
		HHnode/.style={align=center, dotted,rectangle, draw=black!100, fill=white!100, very thick, minimum size=10mm},
]

\node[AB]      (RF source)                        {Alice(IRS)};
\node[pointnode]        (opoint)       [right=2cm of RF source]   {}   ;
\node[RF]        (Alice)       [above=1.5cm of opoint] {RF source};
\node[AB]      (Bob)       [right=2cm of opoint] {Bob};
\node[willie]        (Willie)       [below=1.5cm of opoint] {Willie};
\node[HHnode]      (H0H1)          [right=1cm of Willie] {$H_0:Q_0^n$\\$H_1:Q_1^n$};

\node (01) at ([shift={(-3.5em,-1em)}]Alice.west) {$\textbf{S}_\mathrm{R}$};
\node (02) at ([shift={(2em,1em)}]RF source.east) {$\textbf{S}_\mathrm{R}\textbf{c}$};
\node (03) at ([shift={(-3.5em,1em)}]Willie.west) {$\textbf{S}_\mathrm{R}\textbf{c}$};
\node (04) at ([shift={(1em,1em)}]Willie.north) {$\textbf{S}_\mathrm{R}$};
\node (05) at ([shift={(3em,-1em)}]Alice.east) {$\textbf{S}_\mathrm{R}$};

\draw[-latex,dashed]  (Alice.west)--(RF source.north) ;
\draw[-latex,dashed] (Alice.south) -- (Willie.north);
\draw[-latex,dashed] (Alice.east) -- (Bob.north);
\draw [-latex] (RF source.east) --(Bob.west);
\draw [-latex] (RF source.south) --(Willie.west);
\draw[-latex] (Willie.east) -- (H0H1.west);
\end{tikzpicture}
	\caption{System model for IRS covert communication}
\end{figure}
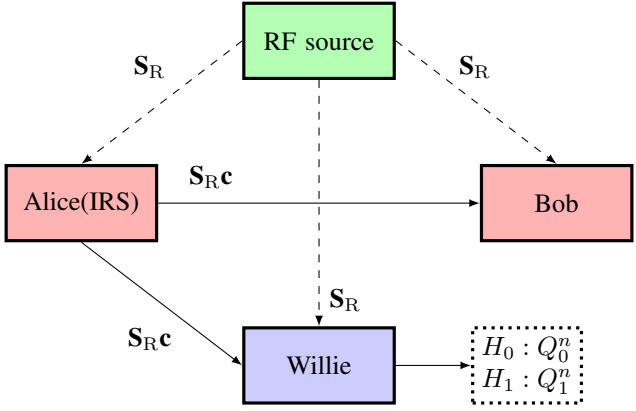

\subsection{Codebook Construction}
\label{Codebook Construction}
Alice transmits covert information to Bob by varying the reflection coefficient of IRS, including the amplitude and phase coefficient. Specifically, Alice transmits a uniformly-distributed message $W\in[\![1,M]\!]$ to Bob by encoding it into a codeword $\mathbf{c}^n=\big[{c}_1,{c}_2,...,{c}_n\big]$ of blocklength $n$. For each symbol $c_i$ of the codeword, we use its amplitude $\|c_i\|$ and phase $\theta_{c_i}$ (instead of its $x$-component $c_{i,x}$ and $y$-component $c_{i,y}$) to describe it. Clearly, we have
\begin{align}
    \|c_i\|^2 &= \|c_{i,x}\|^2 +\|c_{i,y}\|^2,\\
    \tan(\theta_{c_i}) &= c_{i,y} / c_{i,x}.
\end{align}

In the following section, we use $c_i$ and $(\|c_i\|,\theta_{c_i})$ interchangeably when there is no confusion. In this work, we consider three codebooks, and the constructions of the three codebooks are provided as follows:

\subsubsection{The BPSK codebook}
\label{BPSK}
In this codebook, each symbol $c_i$ has the same amplitude $\beta$ and two possible phases $\theta$ and $\theta+\pi$. That is, $c_i$ equals either $(\beta, \theta)$ or $(\beta, \theta+\pi)$. All symbols in this work are constructed in pairs, such as $(\beta, \theta)$ and $(\beta, \theta+\pi)$. For ease of notation, we re-denote the symbol $(\beta, \theta+\pi)$ as $(-\beta, \theta)$, and we also call these two angles, e.g., $\theta$ and $\theta+\pi$, as \textit{one angle pair}. We sample the codeword $\mathbf{c}^n$ independently and randomly according to the distribution $P_{B}^n(\mathbf{c}^n)=\prod_{i=1}^{n}P_{B}(c_i)$ with $P_B(c_i = (\beta, \theta)) = P_B(c_i = (-\beta, \theta)) = 1/2$.

\subsubsection{The 2N-PSK codebook}
\label{NPSK}
In this codebook, each symbol $c_i$ has the same amplitude $\beta$ and $2N$ possible phases $\frac{t\pi}{N}$, $t=1,2,\dots,2N$. Equivalently, each symbol can be described as $c_i = (\beta, \frac{t\pi}{N})$ or $c_i = (-\beta, \frac{t\pi}{N})$, $t=1,2,\dots,N$. Now, we sample the codeword $\mathbf{c}^n$ independently and randomly according to $P_{2N}^n(\mathbf{c}^n)=\prod_{i=1}^{n}P_{2N}(c_i)$ with $P_{2N}(c_i=(\beta,\frac{t\pi}{N}))=P_{2N}(c_i=(-\beta,\frac{t\pi}{N}))=\frac{1}{2N}, t=1,2,...,N$.

\subsubsection{The N-BPSK codebook}
\label{BPSK+}
This codebook is transformed from the BPSK codebook. The phase angle of each codeword is additionally added a uniformly random distributed phase angle. In other words, Alice transmits each codeword in BPSK with an additional phase angle $\widehat{\theta}$, which is selected uniformly at random from $\{\frac{\pi}{N}, \frac{2\pi}{N},...,\frac{(N-1)\pi}{N},\pi\}$. This additional phase angle is shared \textit{confidentially} with Bob in advance.\footnote{\textcolor{black}{In many covert communication scenarios, Alice shares a key with Bob before the communication, typically of size $\mathcal{O}(\sqrt{n})$ bits. So it is reasonable to assume that Alice and Bob share an additional key of size $\mathcal{O}(\log N)$ perfectly and confidentially to reach a consensus on the phase angle.}} Compared to the 2N-PSK codebook, each symbol in a codeword from N-BPSK only has two possible phase angles instead of $2N$ possible phase angles. Specifically, in the 2N-PSK codebook, the phases of different symbols can belong to different angle pairs, e.g., $c_1=(\beta,\frac{\pi}{N})$ and $c_2=(\beta,\frac{2\pi}{N})$. However, in the N-BPSK codebook, the phases of different symbols must belong to one angle pair, e.g., if $c_1=(\beta,\frac{\pi}{N})$, we have $c_2=(\beta,\frac{\pi}{N})$ or $c_2=(-\beta,\frac{\pi}{N})$.

It is assumed that the codebook is revealed to Willie, including the value of amplitude gain $\beta$, and the set of all possible angles, i.e., $\{\theta, \frac{\pi}{N}, \frac{2\pi}{N},...,\frac{(N-1)\pi}{N},\pi\}$.

\subsection{Hypothesis test}
\label{Hypothesis test}
We make a practical assumption that Willie can recover the broadcast information $\mathbf{S_{{R}}}$ and subtract it from his observations to enhance the detection performance. Hence, the $i$-th symbol received by Willie can be equivalently rewritten as
\begin{align}
    &S_{W,i}={{h_{RA}}}{{h_{AW}}}{S_{R,i}c_{i}}+Z_{W,i}.
\end{align}

 Considering a quasi-static flat fading channel with coefficients ${{h_{RA}}}$ and ${{h_{AW}}}$, we define the expected amplitude of complex-value ${{h_{RA}}}{{h_{AW}}}{S_{{R}}}$ as $A$, i.e,
 \begin{align}
   \mathbb{E}(|{{h_{RA}}}{{h_{AW}}}{S_{{R}}}|^2)=A^2,
 \end{align}
and we define the expected phase of ${{h_{RA}}}{{h_{AW}}}{S_{{R}}}$ as ${\theta_0}$. To determine whether Alice is communicating, Willie performs a binary hypothesis test\cite{lehmann2005testing} based on $n$ successive observations $\mathbf{S_{W}}=\{S_{W,i}\}^n_{i=1}$. Specifically, let the null hypothesis ($\mathcal{H}_0$) denote that no communication is taking place, where each sample $S_{W,i}=Z_{W,i}$ is an i.i.d. complex-Gaussian random variable distributed according to $\mathcal{CN}(0,2\sigma^2)$. The alternative hypothesis ($\mathcal{H}_1$) denotes that communication is taking place and each sample $S_{W,i}=h_{RA}h_{AW}S_{R,i}c_i+Z_{W,i}$. Willie aims to distinguish these two hypotheses:
\begin{align}
    \mathcal{H}_0:\quad &S_{W,i}=Z_{W,i},\\
    \mathcal{H}_1:\quad &S_{W,i}={{h_{RA}}}{{h_{AW}}}{S_{R,i}c_{i}}+Z_{W,i}.\label{Eq.5}
\end{align}

Let ${Q}^n_0$ (resp. $\bar{Q}^{(n)}_1$) denote the probability distribution of Willie's $n$ observations when $\mathcal{H}_0$ (resp. $\mathcal{H}_1$) is true. The probability of \emph{false alarm} (i.e., rejecting $\mathcal{H}_0$ when it is true) is denote by $\mathcal{P}_{FA}$, and the probability of \emph{missed detection} (i.e., accepting $\mathcal{H}_0$ when it is false) is denoted by $\mathcal{P}_{MD}$. 
We assume that the distribution ${Q}_0^n$ and $\bar{Q}^{(n)}_1$ are known to Willie, and Willie can perform an optimal statistical hypothesis test that satisfies $\mathcal{P}_{FA}+\mathcal{P}_{MD}=1-\mathbb{V}(\bar{Q}_1^{(n)}\|{Q}_0^{n})$\cite{bash2013limits}. By using the definition of KL divergence and Pinsker’s inequality, we can obtain that the optimal test satisfies $\mathcal{P}_{FA}+\mathcal{P}_{MD}\geq 1-\sqrt{\mathcal{D}(\bar{Q}_1^{(n)}\|{Q}_0^{n})}$. It is possible for Willie to perform a blind test when the sum of error probabilities equal one, i.e., $\mathcal{P}_{FA}+\mathcal{P}_{MD}= 1$. And the objective of covert communication is to guarantee that Willie's statistical test is not much better than blind test. Therefore, we can achieve covert communication amounts to ensuring that $\mathcal{D}(\bar{Q}_1^{(n)}\|{Q}_0^{n})$ is negligible, i.e., ensuring the KL divergence
\begin{align}
    \mathcal{D}(\bar{Q}_1^{(n)}\|{Q}_0^{n})\leq\epsilon,\label{Eq.6}
\end{align}
for an arbitrarily small value $\epsilon\in(0,1)$. Based on the covertness constraint in \eqref{Eq.6}, we now turn to analyze the KL divergences for the three codebooks described in Sec. \ref{Codebook Construction}:
\begin{enumerate}
    \item BPSK: In this case, Alice transmits $(-\beta,\theta)$ or $(\beta,\theta)$ equiprobably. Willie knows how the codebook constructed and the set of all possible angles, and receives the symbols corrupted by AWGN. Therefore, the probability distribution of $\mathbf{S_{W}}$ under $\mathcal{H}_0$ can be expressed as
\begin{align}
    {Q}_0^n(x^n,y^n)=\prod_{i=1}^n\frac{1}{{2\pi}\sigma^2}\exp\bigg({-\frac{x_i^2+y_i^2}{2\sigma^2}}\bigg)\label{Eq.6.5},
\end{align}
and the probability distribution of $\mathbf{S_{W}}$ under $\mathcal{H}_1$ is given by \eqref{def:Q_1} on the top of the next page, where $A$ and ${\theta_0}$ are the expected amplitude and phase of the complex value ${{h_{RA}}}{{h_{AW}}}\mathbf{S_{{R}}}$. For simplicity we use $\mathcal{D}_{B}$ to denote the KL divergence between the distributions ${Q}_1^{n}(x^n,y^n)$ and ${Q}_0^{n}(x^n,y^n)$, i.e., $\mathcal{D}_{B}\triangleq\mathcal{D}({Q}_1^{n}\|{Q}_0^{n})$. 
\begin{figure*}
\begin{align}
    Q_1^n(x^n,y^n)&= \prod_{i=1}^{n}\frac{1}{2}\frac{1}{{2\pi}\sigma^2}\Bigg[\exp{\bigg(-\frac{(x_i+A\beta\cos(\theta_0+{\theta}))^2}{2\sigma^2}-\frac{(y_i+A\beta\sin(\theta_0+{\theta}))^2}{2\sigma^2}\bigg)}\nonumber\\
    &+\exp{\bigg(-\frac{(x_i-A\beta\cos{(\theta_0+{\theta})})^2}{2\sigma^2}-\frac{(y_i-A\beta\sin(\theta_0+{\theta}))^2}{2\sigma^2}\bigg)}\Bigg],\label{def:Q_1}\\
    \widetilde{Q}_1^{n}(x^n,y^n)&= \prod_{i=1}^{n}\frac{1}{2N}\frac{1}{{2\pi}\sigma^2}\sum_{t=1}^N\Bigg[\exp{\bigg(-\frac{(x_i+A\beta\cos{({\theta_0}+\frac{t\pi}{N})})^2}{2\sigma^2}-\frac{(y_i+A\beta\sin{({\theta_0}+\frac{t\pi}{N})})^2}{2\sigma^2}\bigg)}\nonumber\\
    &+\exp{\bigg(-\frac{(x_i-A\beta\cos{({\theta_0}+\frac{t\pi}{N})})^2}{2\sigma^2}-\frac{(y_i-A\beta\sin{({\theta_0}+\frac{t\pi}{N})})^2}{2\sigma^2}\bigg)}\Bigg]\label{Eq.7},\\
    \widehat{Q}_1^{(n)}(x^n,y^n)&= \frac{1}{N}\sum_{t=1}^N\prod_{i=1}^{n}\frac{1}{2}\frac{1}{{2\pi}\sigma^2}\Bigg[\exp{\bigg(-\frac{(x_i+A\beta\cos{({\theta{'}}+\frac{t\pi}{N})})^2}{2\sigma^2}-\frac{(y_i+A\beta\sin{({\theta{'}}+\frac{t\pi}{N})})^2}{2\sigma^2}\bigg)} \nonumber\\
    &+\exp{\bigg(-\frac{(x_i-A\beta\cos{({\theta{'}}+\frac{t\pi}{N})})^2}{2\sigma^2}-\frac{(y_i-A\beta\sin{({\theta{'}}+\frac{t\pi}{N})})^2}{2\sigma^2}\bigg)}\Bigg]\label{Eq.8}.
\end{align}\hrulefill
\end{figure*}
    \item 2N-PSK: In this case, Alice transmits each symbol with amplitude gain $-\beta$ or $\beta$ equiprobably and the phase angle ${\theta}$ is drawn uniformly at random from $\{\frac{\pi}{N}, \frac{2\pi}{N},...,\frac{(N-1)\pi}{N},\pi\}$. The probability distribution of $\mathbf{S_{W}}$ under $\mathcal{H}_0$ is given by $Q_0^n(x^n,y^n)$ stated in \eqref{Eq.6.5}, and the probability distribution of $\mathbf{S_{W}}$ under $\mathcal{H}_1$ can be expressed as \eqref{Eq.7} on the top of the next page. Let $\mathcal{D}_{2N}\triangleq\mathcal{D}(\widetilde{Q}_1^{n}\|{Q}_0^{n})$ denotes the KL divergence between the distributions $\widetilde{Q}_1^{n}(x^n,y^n)$ and ${Q}_0^{n}(x^n,y^n)$.
   \item N-BPSK: When Alice sends covert messages with N-BPSK, she selects a phase angle $\widehat{\theta}$ uniformly at random from $\{\frac{\pi}{N}, \frac{2\pi}{N},...,\frac{(N-1)\pi}{N},\pi\}$, and deflects the $n$ symbols with this selected phase. Then, Alice transmits each symbol with $\beta$ or $-\beta$ equiprobably. The probability distribution of $\mathbf{S_{W}}$ under $\mathcal{H}_0$ is given by $Q_0^n(x^n,y^n)$ stated in \eqref{Eq.6.5}, and the probability distribution of $\mathbf{S_{W}}$ under $\mathcal{H}_1$ is given by \eqref{Eq.8} on the top of the next page, where $\theta'=\theta_0+\theta$. Let $\mathcal{D}_{NB}\triangleq\mathcal{D}(\widehat{Q}_1^{(n)}\|{Q}_0^{n})$ denotes the KL divergence between the distributions $\widehat{Q}_1^{(n)}(x^n,y^n)$ and ${Q}_0^{n}(x^n,y^n)$.
\end{enumerate}

\begin{remark}
It is worth noting that $\widetilde{Q}_1^{n}(x^n,y^n)$ is an $n$-letter product distribution. In other words, the phase of each symbol in $\widetilde{Q}_1^{n}(x^n,y^n)$ is selected independently. However, $\widehat{Q}_1^{(n)}(x^n,y^n)$ is the sum of $N$ distributions, where each distribution in the summation is an $n$-letter product distribution. Further, the phases of all symbols in an inner $n$-letter product distribution belong to one angle pair.
\end{remark}
\textcolor{black}{\begin{remark}The uniform distribution is the classic setting in standard communication scenarios, which provides maximum phase separation between adjacent points and immunity to corruption. By employing a uniform distribution, an analytical expression for the phase gain can be obtained, which facilitates the demonstration of our results. When the generation probability of constellations is not equiprobable, this situation has further been discussed in \cite{Ma2022Optimal}. It should be pointed out that the author of \cite{Ma2022Optimal} gives the optimal probabilistic constellation shaping design of 2N-PSK codebook, which can effectively improve the covert rate.
\end{remark}}

\subsection{Problem Formulation}
\textcolor{black}{We first point out that the connections between the transmission rate and decoding error probability has been established for BPSK \cite{bash2013limits,chen2021on} and for 2N-PSK \cite{Chen2003General,Topal2021ACountermeasure}}.\footnote{Due to the shared phase angle, the N-BPSK codebook can be treated as a normal BPSK codebook from the perspective of decoding.} Thus, it suffices to focus on covertness in the following.

In Sec. \ref{Hypothesis test}, we assume that Willie can construct an optimal statistical hypothesis test. Based on the covertness constraint in \eqref{Eq.6}, we need to ensure that all KL-divergences are less than $\epsilon$, i.e.,
\begin{align}
\max\{\mathcal{D}_{B},\mathcal{D}_{2N},\mathcal{D}_{NB}\}&\leq\epsilon.
\end{align}

In the BPSK codebook, we have no additional improvement on covertness compared to covert communication in real-valued AWGN channel\cite{bash2013limits}. Considering the symmetry of complex-Gaussian noise, it is not hard to understand that different initial phase angles have the same effect on covertness. Thus, the result is not surprising. However, it is interesting to investigate if there is an improvement in covertness performance when each symbol is deflected with different phase angles. To measure the improvement that we achieve from the phase, we take the ratio of the KL divergences between the codebook BPSK and other codebooks as the phase gain. Specifically, the phase gain obtained in the codebook 2N-PSK and N-BPSK are defined as follows:
\begin{definition}
\label{def.1}
For any $0<\epsilon<1$, and for any $\beta>0$ such that $\lim_{n \to \infty } \mathcal{D}_B\leq \epsilon$~\footnote{It is worth noting that $\mathcal{D}_{B}$, $\mathcal{D}_{NB}$ and $\mathcal{D}_{2N}$ are functions of $\beta$.},
\begin{align}
    \alpha_1(N)&=\lim_{n \to \infty } \frac{\mathcal{D}_{2N}}{\mathcal{D}_B},\label{defeq1}\\
    \alpha_2(N)&=\lim_{n \to \infty } \frac{\mathcal{D}_{NB}}{\mathcal{D}_B}.\label{defeq2}
\end{align}
\end{definition}

\section{Main Result}

 \begin{theorem}\label{theorem1}
 For any number of angle pairs $N\geq2$, any covertness parameter $0<\epsilon<1$, let $\beta=\big(\frac{4\epsilon}{n}\big)^{\frac{1}{4}}\frac{\sigma}{A}$, we have
 \begin{align}
 &\mathcal{D}_B=\epsilon+\mathcal{O}(\epsilon^{\frac{3}{2}}n^{-\frac{1}{2}}),\\
     &\alpha_1(N)=\alpha_2(N)=\frac{1}{2}.
 \end{align}
\end{theorem}

Theorem \ref{theorem1} shows that a covertness gain of 2 can be obtained by using either the 2N-PSK codebook or the N-BPSK codebook. Compared with the codebook BPSK with only one phase angle pair, a randomly generated codebook with $N$ possible phase angle pairs (where $N \geq 2$) effectively elevates Willie's uncertainty, and thus improves the covertness. Perhaps more interestingly, when the number of phase angle pairs $N >2$, the covertness gain is ``saturated" i.e., it will no longer increase as $N$ increases. This means that we can achieve a covertness gain of $2$ by simply using $2$ phase angle pairs. 
\textcolor{black}{\begin{remark}In this work, we analyze and compare the performances of BPSK, 2N-PSK, and N-BPSK codebooks with the ratio of KL divergence, aiming to better demonstrate the phase gain. In the BPSK codebook, the initial phase angle can be eliminated by rotating the coordinate system, hence the phase is not utilized. However, this is not feasible when $N\geq2$ as the presence of phase angles cannot be removed through any operation. As a consequent, the phase gain can be obtained when $N\geq 2$. Intuitively, when $N\geq 2$, since the shape distribution of constellations in our work is uniform and symmetric, and the noise is symmetric, increasing the value of $N$ does not increase Willie's uncertainty. Thus the ratio of KL divergence tends to $\frac{1}{2}$ when $n\to\infty$. However, the exact value of KL divergence may change, which can be observed in the simulation section. See section \ref{numerical}.
\end{remark}}
As a consequence, by using the codebook 2N-PSK or N-BPSK, Alice can transmit messages to Bob with $\sqrt{2}$ times the power of the codebook BPSK and simultaneously ensure Willie cannot detect the communication. Moreover, via the codebook 2-BPSK, we can achieve a covertness gain of 2 by just confusing Willie about whether Alice deflects $n$ symbols with a particular phase angle or not. Such a simple operation can yield such significant benefits for achieving covertness. The proof is given in Sec. \ref{Codebook 2} and Sec. \ref{Codebook 3}.

\begin{remark}
The rationale of using $\mathcal{D}(Q_1^n\|Q_0^n)$, rather than $\mathcal{D}(Q_0^n\|Q_1^n)$ (see \cite{bash2013limits} and \cite{chen2021on}), as the covertness metric is as follows. We first point out that both $\mathcal{D}(Q_1^n\|Q_0^n)$ and $\mathcal{D}(Q_0^n\|Q_1^n)$ can be used due to the fact that $\mathbb{V}(Q_1^n\|Q_0^n)=\mathbb{V}(Q_0^n\|Q_1^n)$ and the Pinsker's inequality. However, different metrics yields different optimal signaling over AWGN channels. In some previous information-theoretic works (e.g., \cite{wang2016fundamental}, and \cite{tahmasbi2020covert,lee2018covert,zivarifard2022keyless}), $\mathcal{D}(Q_1^n\|Q_0^n)$ has been used to achieve tight results. Following the aforementioned works, we use $\mathcal{D}(Q_1^n\|Q_0^n)$ as the covertness metric in this work. \textcolor{black}{Moreover, similar to our scenario, the authors in \cite{yan2019Gaussian} proved that the upper bound on $\mathcal{D}(Q_1^n\|Q_0^n)$ is a tighter covertness constraint than that on $\mathcal{D}(Q_0^n\|Q_1^n)$, by proving $\mathcal{D}(Q_0^n\|Q_1^n)\leq \mathcal{D}(Q_1^n\|Q_0^n)$.} 
\end{remark}

\begin{remark} [Technical challenges and solutions]
In some previous works, the chain rule has been widely adopted to transform the KL-divergence of $n$-letter product distributions, e.g., $\mathcal{D}(Q_1^n\|Q_0^n)$, into the sum of $n$ single-letter distributions, e.g., $n\mathcal{D}(Q_1\|Q_0)$\cite{bash2013limits,chen2021on,yan2019Gaussian,Yan2019Delay-Intolerant}. However, for the N-BPSK codebook, the probability distribution $\widehat{Q}_1^{(n)}$ is the average of multiple $n$-letter distributions, so it can not be transformed to single-letter distributions by the chain rule. Additionally, we cannot present a good upper bound on $\mathcal{D}(\widehat{Q}_1^{(n)}\|Q_0^{n})$ via simple inequalities, such as the log-sum inequality, with which we cannot achieve any phase gain. Here we take a different approach to explicitly analyze the KL divergence between ${\widehat{Q}_1^{(n)}}$ and ${Q_0^{n}}$. Firstly, using Taylor series expansion of $\displaystyle \lim_{x \to 0}\log(1+x)$, we can obtain a fact that $\displaystyle \lim_{x \to 0}\log(1+x)=x-\frac{1}{2}x^2+\mathcal{O}(x^3)$. Then, we express the KL divergence as $E_{\widehat{Q}_1^{(n)}}\log{\frac{\widehat{Q}_1^{(n)}}{Q_0^{n}}}$, and we prove that $\frac{\widehat{Q}_1^{(n)}}{Q_0^{n}}$ converges to 1 when $n$ is large enough and $\beta$ is sufficiently small. Following by the fact, we can calculate the first three terms of the expansion of $E_{\widehat{Q}_1^{(n)}}\log{\frac{\widehat{Q}_1^{(n)}}{Q_0^{n}}}$, and obtain an approximation of $\mathcal{D}(\widehat{Q}_1^{(n)}\|Q_0^{n})$ by summing up the approximations of three terms. The proof is given in Sec. \ref{Codebook 3}.
\end{remark}

\begin{remark}

We can achieve a phase gain of 2 when all $N$ phase angle pairs are arithmetic sequence. When $N=2$, i.e., 4-PSK or 2-BPSK, only all possible phase angles are arithmetic sequence, we can achieve a phase gain of 2, and the proof is given in Appendix \ref{section:H}. When $N\to\infty$, we have no specific calculation, but we conjecture that we can achieve the ``max gain" only when all $N$ phase angle pairs are arithmetic sequence.
\end{remark}

\section{Proof of Achievability}
\label{Sec.4}
For ease of notation, we omit all normalization terms, e.g., $\frac{1}{\sqrt{{2\pi}\sigma^2}}$, in the following proof.
\subsection{BPSK}
In this case, by the chain rule of KL-divergence and the fact that both $Q_1^n$ and $Q_0^n$ are $n$-letter product distributions, we have $\mathcal{D}_{B}=n\mathcal{D}(Q_1\|Q_0)$.  With some calculations, we can obtain the approximation of $\mathcal{D}_{B}$ as follows:
\begin{align}
    \mathcal{D}_{B}&=\frac{nA^4\beta^4}{4\sigma^4}-\frac{nA^6\beta^6}{6\sigma^6}+\mathcal{O}(\beta^8).
\end{align}
The proof is given in Appendix \ref{section:A}.

Then, by setting $\beta=\big(\frac{4\epsilon}{n}\big)^{\frac{1}{4}}\frac{\sigma}{A}$, when $n$ is large enough, we can obtain
\begin{align}
     \mathcal{D}_{B}=\epsilon+\mathcal{O}(\epsilon^{\frac{3}{2}}n^{-\frac{1}{2}})\label{Eq.10}.
 \end{align}

\subsection{2N-PSK}
\label{Codebook 2}
In this case, considering $N\geq2$, Alice reflects each symbol with a random phase angle ${\theta}$, where ${\theta}$ is independently drawn from $\{\frac{\pi}{N}, \frac{2\pi}{N},...,\frac{(N-1)\pi}{N},\pi\}$. By the chain rule and the fact that $\widetilde{Q}_1^n$ is a $n$-letter product distribution, we have $\mathcal{D}_{2N}=n\mathcal{D}(\widetilde{Q}_1\|{Q}_0)$. Before we show the calculation of $\mathcal{D}_{2N}$, we need the following lemmas.
\begin{lemma}
\label{Lemma 1}
Considering $N\geq3$, $\theta\in\mathbb{R}$, we can obtain some identities as follows:
\begin{align}
    &\sum_{t=1}^{N}\sin{(\frac{2t\pi}{N}+{\theta})}=\sum_{t=1}^{N}\cos{(\frac{2t\pi}{N}+{\theta})}=0\label{Eq.11},\\
    &\sum_{t=1}^{N}\cos^4({\theta}+\frac{t\pi}{N})=\frac{3N}{8}\label{Eq.12},\\
    &\sum_{t=1}^{N}\sin^2({\theta}+\frac{t\pi }{N})\cos^2({\theta}+\frac{t\pi }{N})=\frac{N}{8}\label{Eq.13},\\
    &\sum_{t=1}^{N}\cos^3({\theta}+\frac{t\pi}{N})\sin({\theta}+\frac{t\pi}{N})=0\label{Eq.14}.
\end{align}
\end{lemma}
\begin{proof}
The proof is given in Appendix \ref{section:B}.
\end{proof}

\begin{lemma}
\label{Lemma 2}
Let us define $\chi_t\triangleq \frac{1}{2}(e^{Y_tB}+e^{-Y_tB})$. Performing Taylor series expansion on $B$ at the point zero, we can obtain the expansion of $\log{\frac{\sum_{t=1}^{N}\chi_t}{N}}$ as follows:
\begin{align}
    &\log{\frac{\sum_{t=1}^{N}\chi_t}{N}}=\frac{B^2}{2!}\bigg(\frac{\sum_{t=1}^{N}{{Y_t}^2}}{N}\bigg)+\frac{B^4}{4!}\cdot\frac{\sum_{t=1}^{N}{{Y_t}^4}}{N}\nonumber\\
    &\qquad-\frac{B^4}{4!}\cdot\frac{3{(\sum_{t=1}^{N}{{Y_t}^2}})\times({\sum_{t=1}^{N}{{Y_t}^2}})}{N^2}+\mathcal{O}(B^6)\label{Eq.1}.
\end{align}
\end{lemma}

Recall that the probability distribution $\widetilde{Q}_{1}(x,y)$ in \eqref{Eq.7}, let 
\begin{align}
    Y_t=\frac{A(x\cos{({\theta_0}+\frac{t\pi}{N})}+y\sin{({\theta_0}+\frac{t\pi}{N}}))}{\sigma^2},\label{Eq.18}
\end{align}
then the term $\log{\frac{\widetilde{Q}_1(x,y)}{Q_0(x,y)}}$ can be rewritten as
\begin{align}
    &\log{\frac{\widetilde{Q}_1(x,y)}{Q_0(x,y)}}=-\frac{A^2\beta^2}{2\sigma^2}\nonumber\\
    &\qquad+\log{\bigg[\frac{1}{2N}\sum_{t=1}^{N}\Big[\exp{\Big(Y_t\beta\Big)}+\exp{\Big(-Y_t\beta\Big)}\Big]\bigg]}\label{Eq.19}.
\end{align}

Applying Lemma \ref{Lemma 2}, we can obtain
\begin{align}
    &\log{\frac{\widetilde{Q}_1}{Q_0}}=\mathcal{O}(\beta^6)-\frac{A^2\beta^2}{2\sigma^2}+\frac{\beta^2}{2!}\Bigg(\frac{\sum_{t=1}^{N}{Y_t^2}}{N}\Bigg)\nonumber\\
    &\quad+\frac{\beta^4}{4!}\Bigg(\frac{\sum_{t=1}^{N}{Y_t^4}}{N}-\frac{3{(\sum_{t=1}^{N}{Y_t^2}})\times({\sum_{t=1}^{N}{Y_t^2}})}{N^2}\Bigg).\label{Eq.20}
\end{align}

\begin{figure*}
\begin{align}
     {\sum_{t=1}^{N}{Y_t^2}}&=\sum_{t=1}^{N}\frac{A^2x^2\Big(\cos^2({\theta_0}+\frac{t\pi}{N})\Big)+A^2y^2\Big(\sin^2({\theta_0}+\frac{t\pi}{N})\Big)+2A^2xy\Big(\sin({\theta_0}+\frac{t\pi}{N})\cos({\theta_0}+\frac{t\pi }{N})\Big)}{\sigma^4}\label{Eq.21}\\
     &=\sum_{t=1}^{N}\frac{A^2x^2\Big(\frac{1+\cos(2t\pi/N+2{\theta_0})}{2}\Big)+A^2y^2\Big(\frac{1-\cos(2t\pi/N+2{\theta_0})}{2}\Big)+A^2xy\Big(\sin(2{\theta_0}+\frac{2t\pi }{N})\Big)}{\sigma^4}\label{Eq.22}\\
     &=\frac{NA^2(x^2+y^2)}{2\sigma^4}\label{Eq.23}.
\end{align}\hrulefill
\end{figure*}

We now calculate the term ${\sum_{t=1}^{N}{Y_t^2}}$ in~\eqref{Eq.21}--\eqref{Eq.23} on the top of the next page, where \eqref{Eq.22} follows by the fact that $2\cos^2{x}=1+\cos{(2x)}$, and \eqref{Eq.23} follows from Lemma \ref{Lemma 1}. Similarly, we can obtain
\begin{align}
     &{\sum_{t=1}^{N}{Y_t^4}}=\frac{3NA^4(x^4+y^4+2x^2y^2)}{8\sigma^8}\label{Eq.24},
\end{align}
where \eqref{Eq.24} follows from Appendix \ref{section:C}. With \eqref{Eq.20}, \eqref{Eq.23} and \eqref{Eq.24}, we have
\begin{align}
    &\log{\frac{\widetilde{Q}_1}{Q_0}}=-\frac{A^2\beta^2}{2\sigma^2}+\frac{\beta^2}{2!}\bigg(\frac{A^2(x^2+y^2)}{2\sigma^4}\bigg)\nonumber\\
    &\qquad-\frac{\beta^4}{4!}\bigg(\frac{3A^4(x^4+y^4+2x^2y^2)}{8\sigma^8}\bigg)+\mathcal{O}(\beta^6)\label{Eq.26}.
\end{align}

Then, we can express the divergence $\mathcal{D}(\widetilde{Q}_1\|{Q}_0)$ as follows:
\begin{align}
    &\mathcal{D}(\widetilde{Q}_1\|{Q}_0)=\iint\widetilde{Q}_1\log{\frac{\widetilde{Q}_1}{Q_0}}\mathrm{d}x\mathrm{d}y\\
    &\quad=-\frac{A^2\beta^2}{2\sigma^2}+\iint\widetilde{Q}_1\Bigg[\frac{\beta^2}{2!}\bigg(\frac{A^2(x^2+y^2)}{2\sigma^4}\bigg)\nonumber\\
    &\quad\ \ -\frac{\beta^4}{4!}\bigg(\frac{3A^4(x^4+y^4+2x^2y^2)}{8\sigma^8}\bigg)+\mathcal{O}(\beta^6)\Bigg]\mathrm{d}x\mathrm{d}y.\label{Eq.28}
\end{align}

With some calculations, we can obtain the approximation of KL-divergence $\mathcal{D}(\widetilde{Q}_1\|{Q}_0)$ as
\begin{align}
    \mathcal{D}(\widetilde{Q}_1\|{Q}_0)=\frac{A^4\beta^4}{8\sigma^4}+\mathcal{O}(\beta^6)\label{Eq.29},
\end{align}
where \eqref{Eq.29} follows by Appendix \ref{section:C}. In summary, the number of phase angle pairs is not relevant to the KL divergence.

By setting $\beta=\big(\frac{4\epsilon}{n}\big)^{\frac{1}{4}}\frac{\sigma}{A}$, when $n$ is large enough, we can obtain
\begin{align}
     \mathcal{D}_{2N}&=\frac{1}{2}\epsilon+\mathcal{O}(\epsilon^{\frac{3}{2}}n^{-\frac{1}{2}}).\label{Eq.30}
 \end{align}
 
 Recall Def.~\ref{def.1} and $\mathcal{D}_{B}=\epsilon+\mathcal{O}(\epsilon^{\frac{3}{2}}n^{-\frac{1}{2}})$ in \eqref{Eq.10}, we can obtain
 \begin{align}
     \alpha_1(N)&=\lim_{n \to \infty } \frac{\mathcal{D}_{2N}}{\mathcal{D}_B}=\frac{1}{2}.
 \end{align}

\subsection{N-BPSK}
\label{Codebook 3}

Recall that the probability distribution  $\widehat{Q}_1^{(n)}(x,y)$ in \eqref{Eq.8}, then re-arranging terms yields:
\begin{align}
    \widehat{Q}_1^{(n)}(x^n,y^n)=\frac{1}{N}\sum_{t=1}^{N}{Q_t^n({\mathbf{x,y}})},\label{Def.wide.hQ}
\end{align}
where ${Q_t^n(\mathbf{x,y})}$ is a $n$-letter product distribution of BPSK with phase angle ${\Theta}_t=\theta'+\frac{t\pi}{N}$ with $t=1,2,..,N$. And the single-letter distribution ${Q_t({x,y})}$ is given by
\begin{align}
    &{Q_t{({x,y})}}\nonumber\\
    &=\frac{1}{2}\Bigg[\exp\bigg(-\frac{(x-A\beta\cos{\Theta}_t)^2+(y-A\beta\sin{\Theta}_t)^2}{2\sigma^2}\bigg)\nonumber\\
    &\quad +\exp\bigg(-\frac{(x+A\beta\cos{\Theta}_t)^2+(y+A\beta\sin{\Theta}_t)^2}{2\sigma^2}\bigg)\Bigg].
\end{align}

According to the definition of the KL divergence, we have 
\begin{align}
    &\mathcal{D}(\widehat{Q}_1^{(n)}\|{Q}_0^n)=\iint {\widehat{Q}_1^{(n)}}\log\frac{{\widehat{Q}_1^{(n)}}}{{Q}_0^n}\mathrm{d}{x}\mathrm{d}{y}\\
    &\quad=\iint \frac{1}{N}\sum_{t=1}^{N}{Q_t^n({\mathbf{x,y}})}\log \frac{1}{N}\sum_{p=1}^{N}\frac{Q_p^n({\mathbf{x,y}})}{Q_0^n({\mathbf{x,y}})}\mathrm{d}{x}\mathrm{d}{y}\label{Eq.36}.
\end{align}

Let us define $\Psi\triangleq\frac{1}{N}\sum_{p=1}^{N}\frac{Q_p^n({\mathbf{x,y}})}{Q_0^n({\mathbf{x,y}})}-1$. Note that, when $\beta$ is sufficiently small, which will be determined later, we can obtain $\Psi=\mathcal{O}(\beta^2)$. Following from Taylor series expansion, when $\Psi\to 0$, we have the fact:
\begin{align}
    \log(1+\Psi)=\Psi-\frac{1}{2}\Psi^2+\mathcal{O}(\Psi^3).\label{Eq.35}
\end{align} 

Combing with \eqref{Eq.35}, we can rewrite \eqref{Eq.36} as
\begin{align}
    \mathcal{D}(\widehat{Q}_1^{(n)}\|{Q}_0^n)&=\iint {\widehat{Q}_1^{(n)}}\log\Big(1+{\Psi}\Big)\mathrm{d}{x}\mathrm{d}{y}\\
    &=\iint {\widehat{Q}_1^{(n)}}\Big[\Psi-\frac{1}{2}\Psi^2+\mathcal{O}(\Psi^3)\Big]\mathrm{d}{x}\mathrm{d}{y}\label{Eq.38}.
\end{align}

With some calculation, when $N\geq2$, we can obtain
\begin{align}
    \iint{\widehat{Q}_1^{(n)}}\cdot{\Psi}~\mathrm{d}{x}\mathrm{d}{y}&=\frac{nA^4\beta^4}{4\sigma^4}+\mathcal{O}(\beta^{8}),\label{Eq.39}\\
    \iint{\widehat{Q}_1^{(n)}}\cdot{\Psi}^2~\mathrm{d}{x}\mathrm{d}{y}&=\frac{nA^4\beta^4}{4\sigma^4}+\frac{nA^6\beta^6}{4\sigma^6}+\mathcal{O}(\beta^8),\label{Eq.40}\\
    \iint{\widehat{Q}_1^{(n)}}\cdot{\Psi}^3~\mathrm{d}{x}\mathrm{d}{y}&=\frac{3nA^6\beta^6}{4\sigma^6}+\mathcal{O}(\beta^8).\label{Eq.41}
\end{align}
The proof is given in Appendix \ref{section:D}.

\begin{remark}
The case that $N=1$ is a special case, we cannot apply Lemma~\ref{Lemma 1} due to $\sin{\frac{\pi}{N}}=0$ in \eqref{Eq.57}. Moreover, when $N=1$, the KL-divergence $\mathcal{D}_{NB}$ can also be obtained via this approximation, i.e., $\displaystyle \lim_{x \to 0}\log(1+x)=x-\frac{1}{2}x^2+\mathcal{O}(x^3)$, as $\mathcal{D}_{NB}=\frac{nA^4\beta^4}{4\sigma^4}+\mathcal{O}(\beta^6)$, and the proof is given in Appendix~\ref{section:G}.
\end{remark}

Combing \eqref{Eq.39}, \eqref{Eq.40} and \eqref{Eq.41}, we can calculate \eqref{Eq.38} as
\begin{align}
    &\mathcal{D}(\widehat{Q}_1^{(n)}\|{Q}_0^n)\nonumber\\
    &=\frac{nA^4\beta^4}{4\sigma^4}+\mathcal{O}(\beta^8)-\frac{nA^4\beta^4}{8\sigma^4}+\mathcal{O}(\beta^6)+\mathcal{O}(\beta^6)\\
    &=\frac{nA^4\beta^4}{8\sigma^4}+\mathcal{O}(\beta^6).
\end{align}

By setting $\beta=\big(\frac{4\epsilon}{n}\big)^{\frac{1}{4}}\frac{\sigma}{A}$, when $n$ is large enough, we can obtain
\begin{align}
     \mathcal{D}_{NB}&=\frac{1}{2}\epsilon+\mathcal{O}(\epsilon^{\frac{3}{2}}n^{-\frac{1}{2}})\label{Eq.51}.
 \end{align}
 
Recall Def.~\ref{def.1} and $\mathcal{D}_{B}=\epsilon+\mathcal{O}(\epsilon^{\frac{3}{2}}n^{-\frac{1}{2}})$ in \eqref{Eq.10}, we can obtain
 \begin{align}
     \alpha_2(N)&=\lim_{n \to \infty } \frac{\mathcal{D}_{NB}}{\mathcal{D}_B}=\frac{1}{2}.
 \end{align}
 
\section{Numerical Results}
 \label{numerical}
  \textcolor{black}{In this section, we provide a performance analysis of the BPSK and 2N-PSK codebooks for different reflection amplitudes. Since KL divergence of N-BPSK can only be calculated in $n$ dimensions, it is difficult to simulate the N-BPSK codebook due to the computational constraint. For the BPSK and 2N-PSK codebooks, we set the expected amplitude of the received signal as $A = 1.2W$, and the power of Gaussian noise as $\sigma^2 = 1W$.}

\begin{figure}[ht]
    \centering
   \includegraphics[scale=0.5]{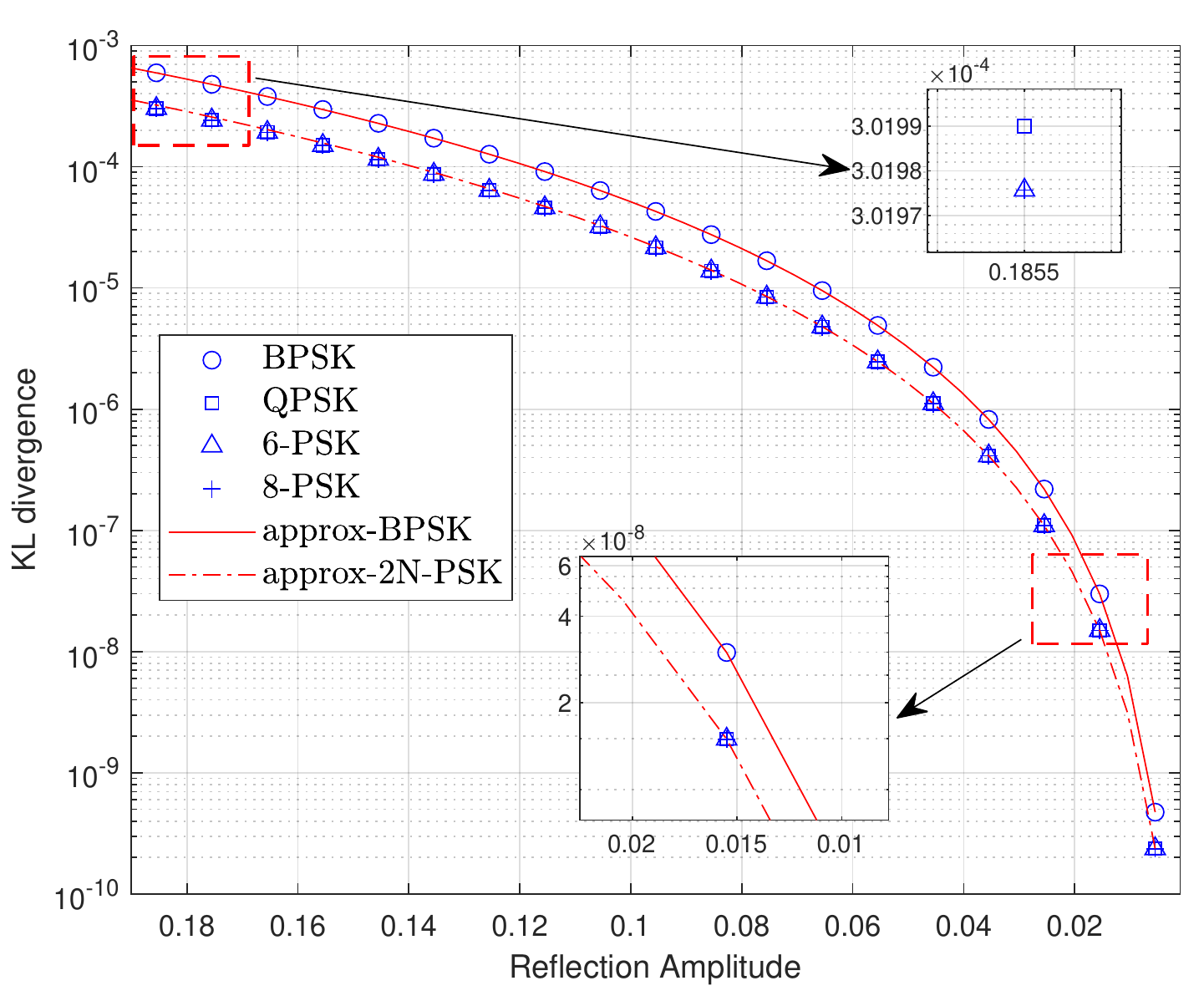}
     \caption{KL divergence versus reflection amplitude}
    \label{fig3}
\end{figure}
\begin{figure}[ht]
    \centering
     \includegraphics[scale=0.5]{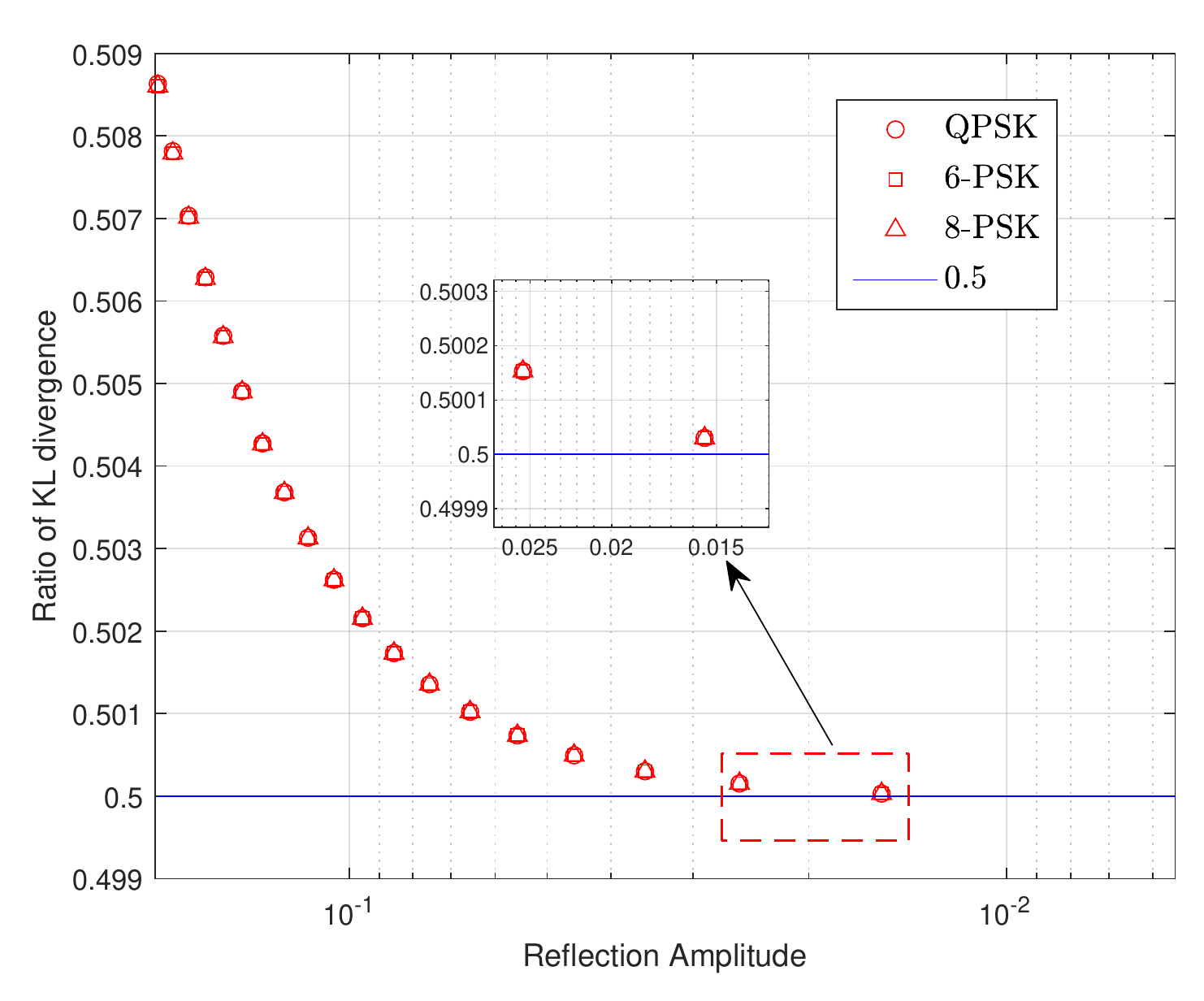}
\caption{Ratio of KL divergence versus reflection amplitude}
\label{fig4}
\end{figure}
\textcolor{black}{Fig.~\ref{fig3} shows the KL divergence of BPSK and 2N-PSK codebooks ($N$ = 2, 3, 4), and provides approximation results, namely $\frac{A^4\beta^4}{4\sigma^4}+\mathcal{O}(\beta^6)$ and $\frac{A^4\beta^4}{8\sigma^4}+\mathcal{O}(\beta^6)$. It can be observed that the 2N-PSK ($N\geq2$) codebook provides better performance (smaller values) than the BPSK codebook, and our approximations are almost consistent with the numerical exact results. Note that as $\beta$ decreases, the approximation becomes more accurate. Moreover, when $\beta$ is not small enough, it can be observed that the KL divergence changes with the number of phase angles $N$.}

\textcolor{black}{Fig.~\ref{fig4} shows the ratio of the KL divergences between 2N-PSK ($N$ = 2, 3, 4) and BPSK codebooks at Willie's side. It is worth noting that, as the reflection amplitude $\beta$ decreases, the ratio of KL divergence tends to our theoretical result $\frac{1}{2}$. This phenomenon supports our main theorem (Theorem 1).}

\section{Conclusion and future work}
This work studies covert communication over complex-valued AWGN channels with intelligent reflect surfaces. We have shown that Alice can achieve a covertness gain up to 2 from phase by leveraging Willie's uncertainty of exact phase angles. Specifically, via the codebooks 2N-PSK and N-BPSK, Alice can transmit messages to Bob with $\sqrt{2}$ times the power of the codebook BPSK while rendering Willie’s detector arbitrarily close to ineffective. Moreover, the covertness gain we achieve from phase will not further increase when the number of phase angle pairs is greater than 2, even if the number approaches infinity. Additionally, via the codebook N-BPSK, we can achieve a covertness gain of 2 from phase by confusing Willie about whether Alice deflects $n$ symbols with a phase angle or not, and such a simple operation can yield a significant covertness benefit.

\textcolor{black}{In this work, we investigate the phase gain in the N-BPSK codebook with a simple deflection, without focusing on finding the optimal scheme to utilize the phase resource. In future work, it would be natural and interesting to study how to utilize phase resources for covert communication optimally.}

\begin{appendices}
\section{}
\label{section:A}
By the chain rule of KL-divergence and the fact that both $Q_1^n$ and $Q_0^n$ are product distributions, we have $\mathcal{D}(Q^n_1\|Q^n_0)=n\mathcal{D}(Q_1\|Q_0)$. We can calculate the KL-divergence between two single-letter distributions as follows:
\begin{align}
    &\mathcal{D} (Q_1\|Q_0 )\nonumber\\
    &={\iint}Q_1(x,y)\log\frac{Q_1(x,y)}{Q_0(x,y)}\mathrm{d}{x}\mathrm{d}{y}\\
    &=-\frac{A^2\beta^2}{2\sigma^2}+{\iint}Q_1(x,y)\nonumber\\
    &\qquad\times\log\bigg[\frac{1}{2}\Big(\exp({\varphi})+\exp({-\varphi})\Big)\bigg] \mathrm{d}{x}\mathrm{d}{y}\\
    &=-\frac{A^2\beta^2}{2\sigma^2}+\bigg[\frac{A^2\beta^2}{2\sigma^2}+\frac{A^4\beta^4}{4\sigma^4}-\frac{A^6\beta^6}{6\sigma^6}+\mathcal{O}(\beta^8)\bigg],
\end{align}
where $\varphi=\frac{A\beta(x\cos(\theta_0+{\theta})+y\sin({\theta_0}+\theta))}{\sigma^2}$, and the last step follows from Taylor series expansion. Then, we can obtain the approximation of $\mathcal{D}(Q_1^n\|Q_0^n)$ as
\begin{align}
    \mathcal{D} (Q_1^n\|Q_0^n)=\frac{nA^4\beta^4}{4\sigma^4}-\frac{nA^6\beta^6}{6\sigma^6}+\mathcal{O}(\beta^8).
\end{align}

\section{Proof of Lemma 1}
\label{section:B}
For any ${\theta}\in\mathbb{R}$, $N\geq3$, we can obtain
\begin{align}
    &\sum_{t=1}^{N}\cos{\Big(\frac{2t\pi}{N}+{\theta}\Big)}\nonumber\\
    &=\frac{1}{\sin{\frac{\pi}{N}}}\sum_{t=1}^{N}\cos\Big({\frac{2t\pi}{N}+{\theta}}\Big)\sin{\frac{\pi}{N}}\label{Eq.57}\\
    &=\frac{1}{2\sin{\frac{\pi}{N}}}\sum_{t=1}^{N}\Bigg[\sin{\bigg({\frac{(2t+1)\pi}{N}+{\theta}}\bigg)}\nonumber\\*
    &\qquad\qquad\qquad\qquad\qquad -\sin{\bigg({\frac{(2t-1)\pi}{N}+{\theta}}\bigg)}\Bigg]\label{Eq.37}\\
    &=\frac{1}{2\sin{\frac{\pi}{N}}}\Bigg[\sin{\bigg({\frac{(2N+1)\pi}{N}+{\theta}}\bigg)}-\sin{\bigg({\frac{\pi}{N}+{\theta}}\bigg)}\Bigg]\\
    &=0.
\end{align}

Similarly, by using the equality that $\sin{(\frac{2t\pi}{N}+{\theta})} = \cos({\frac{2t\pi}{N}+{\theta}}-\frac{\pi}{2})$, after some calculation, we can obtain
\begin{align}
    \sum_{t=1}^{N}\sin{\Big(\frac{2t\pi}{N}+{\theta}\Big)}=0. 
\end{align}

For the rest sine or cos functions with higher degrees, we can reduce the degrees by the double angle formula, e.g., $2\cos^2 x = 1+\cos(2x)$, and calculate similarly as follows:
\begin{align}
    &\sum_{t=1}^{N}\cos^4\bigg({\theta}+\frac{t\pi}{N}\bigg)\nonumber\\
    &=\frac{1}{4}\sum_{t=1}^{N}\bigg(1+2\cos\Big(2{\theta}+\frac{2t\pi}{N}\Big)+\frac{1+\cos\big(4{\theta}+\frac{4t\pi}{N}\big)}{2}\bigg)\\
    &=\frac{3N}{8},
\end{align}
where the last step follows from \eqref{Eq.37}. Similarly, we can obtain
\begin{align}
    &\sum_{t=1}^{N}\sin^2\bigg({\theta}+\frac{t\pi }{N}\bigg)\cos^2\bigg({\theta}+\frac{t\pi }{N}\bigg)\nonumber\\
    &=\frac{1}{4}\sum_{t=1}^{N}\frac{1-\cos\big(4{\theta}+\frac{4t\pi }{N}\big)}{2}\\
    &=\frac{N}{8}\label{Eq.58},
\end{align}
and
\begin{align}
    &\sum_{t=1}^{N}\cos^3\bigg({\theta}+\frac{t\pi}{N}\bigg)\sin\bigg({\theta}+\frac{t\pi}{N}\bigg)\nonumber\\
     &=\frac{1}{8}\sum_{t=1}^{N}\bigg[2\sin\Big(2{\theta}+\frac{2t\pi}{N}\Big)+\sin\Big(4{\theta}+\frac{4t\pi}{N}\Big)\bigg]\\
     &=0\label{Eq.61}.
\end{align}

\section{}
\label{section:C} 
Recall the definition of $Y_t$ in \eqref{Eq.18}, by applying the binomial theory, we can obtain
\begin{align}
     {\sum_{t=1}^{N}{Y_t^4}}
     &=\sum_{t=1}^{N}\frac{A^4\mathbb{X}_t^4+4A^4\mathbb{X}_t^3\mathbb{Y}_t+6A^4\mathbb{X}_t^2\mathbb{Y}_t^2}{\sigma^8}\nonumber\\
     &\qquad +\sum_{t=1}^{N}\frac{4A^4\mathbb{X}_t\mathbb{Y}_t^3+A^4\mathbb{Y}_t^4}{\sigma^8},
\end{align}
where $\mathbb{X}_t=x\cos{({\theta_0}+\frac{t\pi}{N})}$, $\mathbb{Y}_t=y\sin({\theta_0}+\frac{t\pi}{N})$. Similar to \eqref{Eq.23}, after some algebraic calculation and applying Lemma \ref{Lemma 1}, we have
\begin{equation}
   {\sum_{t=1}^{N}{Y_t^4}} = \frac{3NA^4(x^4+y^4+2x^2y^2)}{8\sigma^8}\label{Eq.74}.
\end{equation}

Furthermore, in \eqref{Eq.28}, we have calculated the KL-divergence $\mathcal{D}(\widetilde{Q}_1\|{Q}_0)$ as
\begin{align}
    &{\iint}\widetilde{Q}_1\log{\frac{\widetilde{Q}_1}{Q_0}}\mathrm{d}x\mathrm{d}y\nonumber\\
    &=-\frac{A^2\beta^2}{2\sigma^2}+{\iint}\widetilde{Q}_1\bigg[\mathcal{O}(\beta^6)+\frac{\beta^2}{2!}\Big(\frac{A^2(x^2+y^2)}{2\sigma^4}\Big)\nonumber\\
    &\qquad -\frac{\beta^4}{4!}\Big(\frac{3A^4(x^4+y^4+2x^2y^2)}{8\sigma^8}\Big)\bigg]\mathrm{d}x\mathrm{d}y.
\end{align}

We first calculate the term ${\iint}{\widetilde{Q}_1}\Big[\frac{\beta^2}{2!}\frac{A^2(x^2+y^2)}{2\sigma^4}\Big]\mathrm{d}x\mathrm{d}y$ as follows
\begin{align}
    &{\iint}{\widetilde{Q}_1}\bigg[\frac{\beta^2}{2!}\frac{A^2(x^2+y^2)}{2\sigma^4}\bigg]\mathrm{d}x\mathrm{d}y\nonumber\\
    &=\frac{\beta^2}{2!}\frac{A^2}{2\sigma^4}\frac{1}{N}\bigg(\sum_{t=1}^{N}\bigg[A^2\beta^2\cos^2\Big({\theta_0}+\frac{t\pi}{N}\Big)+\sigma^2\bigg]\nonumber\\
    & \qquad+\sum_{t=1}^{N}\bigg[A^2\beta^2\sin^2\Big({\theta_0}+\frac{t\pi}{N}\Big)+\sigma^2\bigg]\bigg)\\
    &=\frac{A^2\beta^2}{2\sigma^2}+\frac{A^4\beta^4}{4\sigma^4},\label{Eq.69}
\end{align}
where \eqref{Eq.69} follows by Lemma \ref{Lemma 1}. 

Then, the second term $\iint{\widetilde{Q}_1}\Big[\frac{\beta^4}{4!}\Big(\frac{3x^4+3y^4+6x^2y^2}{8\sigma^8}\Big)\Big]\mathrm{d}x\mathrm{d}y$ is calculated in \eqref{app-C:1}--\eqref{app-C:2} on the top of the next page, where \eqref{Eq.72} follows by Lemma \ref{Lemma 1}.

\begin{figure*}
    \begin{align}
    &{\iint}{\widetilde{Q}_1}\bigg[\frac{\beta^4}{4!}\Big(\frac{3x^4+3y^4+6x^2y^2}{8\sigma^8}\Big) \bigg]\mathrm{d}x\mathrm{d}y\nonumber\\
     &=\frac{\beta^4}{4!}\Big(\frac{3}{8\sigma^8}\Big)\Bigg\{\frac{1}{N}\sum_{t=1}^{N}\bigg[A^4\beta^4\cos^4{\Big({\theta_0}+\frac{t\pi}{N}\Big)}+6A^2\beta^2\sigma^2\cos^2\Big({\theta_0}+\frac{t\pi}{N}\Big)+3\sigma^4 \bigg]\nonumber\\
    &\qquad+\frac{1}{N}\sum_{t=1}^{N}\bigg[A^4\beta^4\sin^4{\Big({\theta_0}+\frac{t\pi}{N}\Big)}+6A^2\beta^2\sigma^2\sin^2\Big({\theta_0}+\frac{t\pi}{N}\Big)+3\sigma^4 \bigg]\nonumber\\
    &\qquad\qquad+2\frac{1}{N}\sum_{t=1}^{N}\bigg[\bigg(A^2\beta^2\cos^2{\Big({\theta_0}+\frac{t\pi}{N}\Big)}+\sigma^2 \bigg)\bigg(A^2\beta^2\sin^2{\Big({\theta_0}+\frac{t\pi}{N}\Big)}+\sigma^2 \bigg) \bigg]\Bigg\}\label{app-C:1}\\
     &=\frac{A^4\beta^4}{4!}\Big(\frac{3}{8\sigma^8}\Big)\bigg[6\sigma^4+2\sigma^4+\mathcal{O}({\beta^2}) \bigg]\label{Eq.72}\\
     &=\frac{A^4\beta^4}{8\sigma^4}+\mathcal{O}(\beta^6),\label{app-C:2}
\end{align}\hrulefill
\end{figure*}

\section{}
\label{section:D}
Following by \eqref{Eq.38}, we first calculate the first term $\iint {\widehat{Q}_1^{(n)}}{\Psi}\mathrm{d}x\mathrm{d}y$ as follows
\begin{align}
&{\iint}{\widehat{Q}_1^{(n)}}{\Psi}\mathrm{d}x\mathrm{d}y\nonumber\\
     &=\frac{1}{N}\sum_{t=1}^{N}{\iint}{Q_t^n({\mathbf{x,y}})}\times\bigg(\sum_{p=1}^{N}\frac{1}{N}\frac{Q_p^n({\mathbf{x,y}})}{Q_0^n({\mathbf{x,y}})}-1 \bigg)\mathrm{d}x\mathrm{d}y\\
     &=\frac{1}{N^2}\sum_{t=1}^{N}\sum_{p=1}^{N}{\iint}{Q_t^n({\mathbf{x,y}})}\frac{Q_p^n({\mathbf{x,y}})}{Q_0^n({\mathbf{x,y}})}\mathrm{d}x\mathrm{d}y-1\\
     &=\frac{1}{N^2}\sum_{t=1}^{N}\sum_{p=1}^{N}\prod_{i=1}^n{\iint}{Q_t({\mathbf{x,y}})}\frac{Q_p({\mathbf{x,y}})}{Q_0({\mathbf{x,y}})}\mathrm{d}x\mathrm{d}y-1\\
    &=\frac{1}{N^2}\sum_{t=1}^{N}\sum_{p=1}^{N} \bigg[\frac{1}{2}\exp\bigg({\frac{A^2\beta^2\cos{\Delta_{pt}}}{\sigma^2}}\bigg)\nonumber\\
    &\qquad\qquad\qquad+\frac{1}{2}\exp\bigg({-\frac{A^2\beta^2\cos{\Delta_{pt}}}{\sigma^2}}\bigg) \bigg]^n -1\label{Eq.76}\\
    &=\frac{1}{N^2}\sum_{t=1}^{N}\sum_{p=1}^{N}\bigg(1+\frac{nA^4\beta^4\cos^2{\Delta}_{pt}}{2\sigma^4}+\mathcal{O}(\beta^{8}) \bigg)-1\label{Eq.77}\\
    &=\frac{nA^4\beta^4}{4\sigma^4}+\mathcal{O}(\beta^{8})\label{Eq.78},
\end{align}
where $\Delta_{pt}=({\Theta}_p-{\Theta}_t)$, and the proof of \eqref{Eq.76} is provided in \eqref{app-D:1}--\eqref{app-D:2} on the top of the page after next, and \eqref{Eq.77} follows from Taylor series expansion which is given by

\begin{figure*}
\begin{align}
    &{\iint}{Q_t({\mathbf{x,y}})}\frac{Q_p({\mathbf{x,y}})}{Q_0({\mathbf{x,y}})}\mathrm{d}x\mathrm{d}y\nonumber\\
    &=\frac{1}{4}{\iint}\bigg(\exp\bigg({-\frac{(x-A\beta\cos{{\Theta}}_{t})^2+(y-A\beta\sin{{\Theta}}_{t})^2}{2\sigma^2}}\bigg)+\exp\bigg({-\frac{(x+A\beta\cos{{\Theta}}_{t})^2+(y+A\beta\sin{{\Theta}}_{t})^2}{2\sigma^2}}\bigg) \bigg)\nonumber\\
    &\qquad\times\bigg[\exp\bigg({-\frac{-2A\beta(x\cos({{\Theta}}_{t}+\Delta_{tp})+y\sin({{\Theta}}_{t}+\Delta_{tp}))+A^2\beta^2}{2\sigma^2}}\bigg)\nonumber\\
    &\qquad\qquad +\exp\bigg({-\frac{2A\beta(x\cos({{\Theta}}_{t}+\Delta_{tp})+y\sin({{\Theta}}_{t}+\Delta_{tp}))+A^2\beta^2}{2\sigma^2}}\bigg)\bigg]\mathrm{d}x\mathrm{d}y\label{app-D:1}\\
    &=\frac{1}{4}{\iint}\bigg[\exp\bigg({-\frac{(x-A\beta(\cos{{{\Theta}}_{t}}+\cos{({{\Theta}}_{t}+\Delta_{tp})}))^2+(y-A\beta(\sin{{{\Theta}}_{t}}+\sin{({{\Theta}}_{t}+\Delta_{tp})}))^2}{2\sigma^2}} \bigg) \nonumber\\
    &\qquad +\exp\bigg({-\frac{(x+A\beta(\cos{{{\Theta}}_{t}}+\cos{({{\Theta}}_{t}+\Delta_{tp})}))^2+(y+A\beta(\sin{{{\Theta}}_{t}}+\sin{({{\Theta}}_{t}+\Delta_{tp})}))^2}{2\sigma^2}}\bigg) \bigg]\nonumber\\
    &\qquad\quad\times\exp\bigg({-\frac{2A^2\beta^2-2A^2\beta^2(1+\cos{{\Theta}}_{t}\cos{({{\Theta}}_{t}+\Delta_{tp})+\sin{{\Theta}}_{t}\sin{({{\Theta}}_{t}+\Delta_{tp}))}}}{2\sigma^2}}\bigg)\nonumber\\
    &\qquad\quad\quad+\bigg[\exp\bigg({-\frac{(x-A\beta(\cos{{{\Theta}}_{t}}-\cos{({{\Theta}}_{t}+\Delta_{tp})}))^2+(y-A\beta(\sin{{{\Theta}}_{t}}-\sin{({{\Theta}}_{t}+\Delta_{tp})}))^2}{2\sigma^2}}\bigg)  \nonumber\\
    &\qquad\quad\quad\quad +\exp\bigg({-\frac{(x+A\beta(\cos{{{\Theta}}_{t}}-\cos{({{\Theta}}_{t}+\Delta_{tp})}))^2+(y+A\beta(\sin{{{\Theta}}_{t}}-\sin{({{\Theta}}_{t}+\Delta_{tp})}))^2}{2\sigma^2}}\bigg) \bigg]\nonumber\\
    &\qquad\quad\quad\quad\quad \times\exp\bigg({-\frac{2A^2\beta^2-2A^2\beta^2(1-\cos{{\Theta}}_{t}\cos{({{\Theta}}_{t}+\Delta_{tp})-\sin{{\Theta}}_{t}\sin{({{\Theta}}_{t}+\Delta_{tp}))}}}{2\sigma^2}}\bigg)\mathrm{d}x\mathrm{d}y\\
     &=\frac{1}{4}\bigg[2\exp\bigg({-\frac{2A^2\beta^2-A^2\beta^2(2+2\cos{\Delta}_{tp})}{2\sigma^2}}\bigg)+2\exp\bigg({\frac{2A^2\beta^2-A^2\beta^2(2-2\cos{\Delta}_{tp})}{2\sigma^2}}\bigg) \bigg]\\
     &=\frac{1}{2}\bigg[\exp\bigg({-\frac{A^2\beta^2(\cos{\Delta}_{tp})}{\sigma^2}}\bigg)+\exp\bigg({\frac{A^2\beta^2(\cos{\Delta}_{tp})}{\sigma^2}}\bigg) \bigg],\label{app-D:2}
\end{align}\hrulefill
\end{figure*}

\begin{align}
    &\bigg(\frac{1}{2}\exp\bigg({\frac{A^2\beta^2\cos{\Delta_{pt}}}{\sigma^2}}\bigg)+\frac{1}{2}\exp\bigg({-\frac{A^2\beta^2\cos{\Delta_{pt}}}{\sigma^2}}\bigg)\bigg)^n\nonumber\\
     &= 1+{\frac{nA^4\beta^4\cos^2{\Delta_{pt}}}{2\sigma^4}}\nonumber\\
     &\qquad\qquad-\frac{A^8\beta^8\cos^4{{\Delta_{pt}}}(2n-3n^2)}{24\sigma^8}+\mathcal{O}(\beta^{12})\label{Eq.83},
\end{align}

Further, the proof of \eqref{Eq.78} is given by
\begin{align}
    &\frac{1}{N^2}\sum_{t=1}^{N}\bigg(\sum_{p=1}^{N}\Big(\frac{nA^4\beta^4\cos^2{\Delta}_{pt}}{2\sigma^4}\Big) \bigg)\\
    &=\frac{nA^4\beta^4}{2\sigma^4}\frac{1}{N^2}\sum_{t=1}^{N}\sum_{p=1}^{N}\cos^2\bigg(\frac{(p-t)\pi}{N} \bigg)\\
    &=\frac{nA^4\beta^4}{2\sigma^4}\frac{1}{N^2}\sum_{t=1}^{N}\sum_{p=1}^{N}\frac{1+\cos{\frac{2(p-t)\pi}{N}}}{2}\\
    &=\frac{nA^4\beta^4}{4\sigma^4} + \frac{nA^4\beta^4}{4\sigma^4}\frac{1}{N^2} \frac{1}{\sin{\frac{\pi}{N}}}\sum_{t=1}^{N}\sum_{p=1}^{N}\nonumber\\
    &\qquad\frac{1}{2}\bigg(\sin{\frac{(2p+1)\pi-2t}{N}}-\sin{\frac{(2p-1)\pi-2t}{N}}\bigg)\\
    &=\frac{nA^4\beta^4}{4\sigma^4}.\label{Eq.87}
\end{align}

Now we calculate the second term of the expansion of $\mathcal{D}(\widehat{Q}_1^{(n)}\|{Q}_0^n)$. Following by \eqref{Eq.38}, we can rearrange the second term $ {\iint} {\widehat{Q}_1^{(n)}}{\Psi}^2\mathrm{d}x\mathrm{d}y$ as follows
\begin{align}
    {\iint} {\widehat{Q}_1^{(n)}}{\Psi}^2\mathrm{d}x\mathrm{d}y&={\iint}\Big({\widehat{Q}_1^{(n)}}({\Psi}+1)^2\nonumber\\
    &\qquad -2{\widehat{Q}_1^{(n)}}({\Psi}+1)+{\widehat{Q}_1^{(n)}} \Big)\mathrm{d}x\mathrm{d}y\label{Eq.93}.
\end{align}

Then, recall the definition of $\widehat{Q}_1^{(n)}(x^n,y^n)$ in \eqref{Def.wide.hQ}, we calculate the first term of \eqref{Eq.93} as follows
\begin{align}
     &{\iint} {\widehat{Q}_1^{(n)}}({\Psi}+1)^2\mathrm{d}x\mathrm{d}y\nonumber\\
      &={\iint} {\widehat{Q}_1^{(n)}}\bigg(\sum_{p=1}^{N}\frac{1}{N}\frac{Q_p^n({\mathbf{x,y}})}{Q_0^n({\mathbf{x,y}})} \bigg)^2\mathrm{d}x\mathrm{d}y\\
      &={\iint} \frac{1}{N}\sum_{t=1}^{N}{Q_t^n({\mathbf{x}})}\times\sum_{p=1}^{N}\frac{1}{N}\frac{Q_p^n({\mathbf{x,y}})}{Q_0^n({\mathbf{x,y}})}\nonumber\\
      &\qquad\qquad\times\sum_{z=1}^{N}\frac{1}{N}\frac{Q_z^n({\mathbf{x,y}})}{Q_0^n({\mathbf{x,y}})}\mathrm{d}x\mathrm{d}y\\
     &= \frac{1}{N^3}\sum_{t=1}^{N}\sum_{p=1}^{N}\sum_{z=1}^{N}\prod_{i=1}^{n}{\iint}{Q_t({{x,y}})}\frac{{Q_z({{x,y}})}Q_p({{x,y}})}{Q_0^2({{x,y}})}\mathrm{d}x\mathrm{d}y\\
    &=1+\frac{3nA^4\beta^4}{4\sigma^4}+\frac{nA^6\beta^6}{4\sigma^6}+\mathcal{O}(\beta^8)\label{Eq.98},
\end{align}
and the proof of \eqref{Eq.98} is provided in \eqref{app-D:b1}--\eqref{eq:qq1} on the top of the next page and \eqref{eq:qq2}--\eqref{Eq.104} on the top of the page after next, where $\Delta_{tz}=({\Theta}_t-{\Theta}_z)$.

\begin{figure*}
\begin{align}
    &\frac{1}{N^3}\sum_{t=1}^{N}\sum_{p=1}^{N}\sum_{z=1}^{N}\prod_{i=1}^{n}{\iint}{Q_t({{x,y}})}\frac{{Q_z({{x,y}})}Q_p({{x,y}})}{Q_0^2({{x,y}})}\mathrm{d}x\mathrm{d}y\nonumber\\
    &=\frac{1}{N^3}\sum_{t=1}^{N}\sum_{p=1}^{N}\sum_{z=1}^{N}\prod_{i=1}^{n}\frac{1}{4}\bigg[\exp\bigg({-\frac{A^2\beta^2(-\cos{(\Delta_{tp}-\Delta_{ts})-\cos{\Delta_{tp}}-\cos{\Delta_{ts}}})}{\sigma^2}}\bigg)  \nonumber\\
    &\qquad+\exp\bigg({-\frac{A^2\beta^2(-\cos{(\Delta_{tp}-\Delta_{ts})+\cos{\Delta_{tp}}+\cos{\Delta_{ts}}})}{\sigma^2}}\bigg)+\exp\bigg({-\frac{A^2\beta^2(\cos{(\Delta_{tp}-\Delta_{ts})+\cos{\Delta_{tp}}-\cos{\Delta_{ts}}})}{\sigma^2}}\bigg)\nonumber\\
    & \qquad\qquad +\exp\bigg({-\frac{A^2\beta^2(\cos{(\Delta_{tp}-\Delta_{ts})-\cos{\Delta_{tp}}+\cos{\Delta_{ts}}})}{\sigma^2}}\bigg) \bigg]\label{app-D:b1}\\
    &= \frac{1}{N^3}\sum_{t=1}^{N}\sum_{p=1}^{N}\sum_{z=1}^{N}\Bigg[1+\frac{nA^4\beta^4}{2\sigma^4}+\Big(\frac{2nA^4\beta^4}{2\sigma^4}+\frac{nA^6\beta^6}{2\sigma^6}\Big)(\cos^2{\Delta_{tz}}\cos^2{\Delta_{tp}}+\cos{\Delta_{tz}}\cos{\Delta_{tp}}\sin{\Delta_{tz}}\sin{\Delta_{tp}})+\mathcal{O}(\beta^8)\Bigg].\label{eq:qq1}
\end{align}\hrulefill
\end{figure*}
\begin{figure*}
\begin{align}
    &\frac{1}{N^3}\sum_{t=1}^{N}\sum_{p=1}^{N}\sum_{z=1}^{N}(\cos^2{\Delta_{tz}}\cos^2{\Delta_{tp}}+\cos{\Delta_{tz}}\cos{\Delta_{tp}}\sin{\Delta_{tz}}\sin{\Delta_{tp}})\nonumber\\
     &=\frac{1}{N^3}\sum_{t=1}^{N}\sum_{p=1}^{N}\sum_{z=1}^{N}(\cos^2{\Delta_{tz}}\cos^2{\Delta_{tp}})+\frac{1}{4}\frac{1}{N^3}\sum_{t=1}^{N}\sum_{p=1}^{N}\sum_{z=1}^{N}(\sin{2\Delta_{tz}}\sin{2\Delta_{tp}})\label{eq:qq2}\\
     &=\frac{1}{N^3}\sum_{t=1}^{N}\bigg(\sum_{p=1}^{N}\cos^2{\Delta_{tp}} \bigg)\bigg(\sum_{z=1}^{N}\cos^2{\Delta_{tz}} \bigg)+\frac{1}{4}\frac{1}{N^3}\sum_{t=1}^{N}\sum_{p=1}^{N}\sum_{z=1}^{N}\bigg(\sin{\frac{2\pi(t-z)}{N}}\sin{\frac{2\pi(t-p)}{N}} \bigg)\\
     &= \frac{1}{N^3}\sum_{t=1}^{N}\bigg(\sum_{p=1}^{N}\cos^2{\Delta_{tp}} \bigg)\bigg(\sum_{z=1}^{N}\cos^2{\Delta_{tz}} \bigg)+
     \frac{1}{2\sin{\frac{\pi}{N}}}\sum_{t=1}^{N}\sum_{p=1}^{N}\sin\bigg({\frac{2\pi(t-p)}{N}} \bigg)\nonumber\\
    &\qquad\qquad\times\sum_{z=1}^{N}\bigg[\sin\bigg({\frac{(-2z+1)\pi}{N}}+\frac{(4t-N)\pi}{2N} \bigg)-\sin\bigg({\frac{(-2z-1)\pi}{N}}+\frac{(4t-N)\pi}{2N} \bigg)\bigg]\\
    &=\frac{1}{2}\times\frac{1}{2}+0 = \frac{1}{4}. \label{Eq.104}
\end{align}\hrulefill
\end{figure*}

Combing \eqref{Eq.78}, \eqref{Eq.93} and \eqref{Eq.98}, the term ${\iint} {\widehat{Q}_1^{(n)}}{\Psi}^2\mathrm{d}x\mathrm{d}y$ can be calculated as follows
\begin{align}
    &{\iint} {\widehat{Q}_1^{(n)}}{\Psi}^2\mathrm{d}x\mathrm{d}y\nonumber\\
    &=1+\frac{3nA^4\beta^4}{4\sigma^4}+\frac{nA^6\beta^6}{4\sigma^6}-2\bigg(\frac{nA^4\beta^4}{4\sigma^4}+1\bigg)\nonumber\\
    &\qquad+1+\mathcal{O}(\beta^8)\\
    &=\frac{nA^4\beta^4}{4\sigma^4}+\frac{nA^6\beta^6}{4\sigma^6}+\mathcal{O}(\beta^8).
\end{align}

Finally, we calculate the third term of the expansion of $\mathcal{D}(\widehat{Q}_1^{(n)}\|{Q}_0^n)$. Re-arranging the third term ${\iint} {\widehat{Q}_1^{(n)}}{\Psi}^3$ yields:
\begin{align}
    &{\iint} {\widehat{Q}_1^{(n)}}{\Psi}^3\mathrm{d}x\mathrm{d}y\nonumber\\
    &={\iint} \bigg({\widehat{Q}_1^{(n)}}({\Psi}+1)^3-3{\widehat{Q}_1^{(n)}}({\Psi}+1)^2\nonumber\\
    &\qquad+3{\widehat{Q}_1^{(n)}}({\Psi}+1)-{\widehat{Q}_1^{(n)}} \bigg)\mathrm{d}x\mathrm{d}y.
\end{align}

Due to the limit of the space, after some calculation, we can obtain 
\begin{align}
    &{\iint} {\widehat{Q}_1^{(n)}}({\Psi}+1)^3\mathrm{d}x\mathrm{d}y\nonumber\\
    &=1+\frac{3nA^4\beta^4}{2\sigma^4}+\frac{3nA^6\beta^6}{4\sigma^6}+\mathcal{O}(\beta^8)\label{Eq.118}.
\end{align}

With \eqref{Eq.78}, \eqref{Eq.98} and \eqref{Eq.118}, we can calculate the term ${\iint} {\widehat{Q}_1^{(n)}}{\Psi}^3\mathrm{d}x\mathrm{d}y$ as follows
\begin{align}
    &{\iint} {\widehat{Q}_1^{(n)}}{\Psi}^3\mathrm{d}x\mathrm{d}y\nonumber\\
    &={\iint} \bigg({\widehat{Q}_1^{(n)}}({\Psi}+1)^3-3{\widehat{Q}_1^{(n)}}({\Psi}+1)^2\nonumber\\
    &\qquad+3{\widehat{Q}_1^{(n)}}({\Psi}+1)-{\widehat{Q}_1^{(n)}} \bigg)\mathrm{d}x\mathrm{d}y\\
    &=\bigg(1+\frac{3nA^4\beta^4}{2\sigma^4}+\frac{3nA^6\beta^6}{4\sigma^6}\bigg)\nonumber\\*
    &\qquad-3\bigg(1+\frac{3nA^4\beta^4}{4\sigma^4}+\frac{nA^6\beta^6}{4\sigma^6}\bigg)\nonumber\\
    &\qquad\qquad +3\bigg(1+\frac{nA^4\beta^4}{4\sigma^4}\bigg)-1+\mathcal{O}(\beta^8)\\
    &=\frac{3nA^6\beta^6}{4\sigma^6}+\mathcal{O}(\beta^8).
\end{align}

\section{Proof of N=1}
\label{section:G}
Note that when $N=1$, the term $\Delta_{tp}$ in \eqref{Eq.77} equals zero, then the term ${\iint} {\widehat{Q}_1^{(n)}}{\Psi}\mathrm{d}x\mathrm{d}y$ can be calculated as
\begin{align}
    &{\iint} {\widehat{Q}_1^{(n)}}{\Psi}\mathrm{d}x\mathrm{d}y\nonumber\\
    &=\frac{1}{N^2}\sum_{t=1}^{N}\sum_{p=1}^{N}\bigg(1+\frac{nA^4\beta^4\cos^2{\Delta}_{pt}}{2\sigma^4} \bigg)-1+\mathcal{O}(\beta^{8})\\
    &=\frac{nA^4\beta^4}{2\sigma^4}+\mathcal{O}(\beta^{8}).
\end{align}

Similarly, when $N=1$, we can calculate ${\iint} {\widehat{Q}_1^{(n)}}{\Psi}^2$ as follows
\begin{align}
    &{\iint} {\widehat{Q}_1^{(n)}}{\Psi}^2\mathrm{d}x\mathrm{d}y\nonumber\\
    &={\iint} \bigg({\widehat{Q}_1^{(n)}}({\Psi}+1)^2-2{\widehat{Q}_1^{(n)}}({\Psi}+1)+{\widehat{Q}_1^{(n)}} \bigg)\mathrm{d}x\mathrm{d}y\\
    &=1+\frac{3nA^4\beta^4}{2\sigma^4}+\frac{nA^6\beta^6}{\sigma^6}-2\bigg(\frac{nA^4\beta^4}{2\sigma^4}+1\bigg)+1+\mathcal{O}(\beta^8)\\
    &=\frac{nA^4\beta^4}{2\sigma^4}+\frac{nA^6\beta^6}{\sigma^6}+\mathcal{O}(\beta^8),
\end{align}
and we can calculate ${\iint} {\widehat{Q}_1^{(n)}}{\Psi}^3$ as follows
\begin{align}
    &{\iint} {\widehat{Q}_1^{(n)}}{\Psi}^3\mathrm{d}x\mathrm{d}y\nonumber\\
    &={\iint} \bigg({\widehat{Q}_1^{(n)}}({\Psi}+1)^3-3{\widehat{Q}_1^{(n)}}({\Psi}+1)^2\nonumber\\
    &\qquad+3{\widehat{Q}_1^{(n)}}({\Psi}+1)-{\widehat{Q}_1^{(n)}} \bigg)\mathrm{d}x\mathrm{d}y\\
    &=\bigg(1+\frac{3nA^6\beta^6}{\sigma^6}\bigg)-3\bigg(1+\frac{3nA^4\beta^4}{2\sigma^4}+\frac{nA^6\beta^6}{\sigma^6}\bigg)\nonumber\\
    &\qquad+3\bigg(1+\frac{nA^4\beta^4}{2\sigma^4}\bigg)-1+\mathcal{O}(\beta^8)\\
    &=\frac{3nA^6\beta^6}{\sigma^6}+\mathcal{O}(\beta^8).
\end{align}

Then, we have
\begin{align}
        &\mathcal{D}(\widehat{Q}_1^{(n)}\|{Q}_0^n)\nonumber\\
        &={\iint}{\widehat{Q}_1^{(n)}}\bigg[{\Psi}-\frac{1}{2}{\Psi}^2+{o}({\Psi}^3) \bigg]\mathrm{d}x\mathrm{d}y\\
        &=\bigg(\frac{nA^4\beta^4}{2\sigma^4}\bigg)-\frac{1}{2}\bigg(\frac{nA^4\beta^4}{2\sigma^4}+\mathcal{O}(\beta^6)\bigg)+\mathcal{O}(\beta^6)\\
    &=\frac{nA^4\beta^4}{4\sigma^4}+\mathcal{O}(\beta^6).
\end{align}

\begin{figure*}
    \begin{align}
    \widetilde{Q}_1^{\prime n}(x^n,y^n)&=\prod_{i=1}^{n}\frac{1}{4}\frac{1}{{2\pi}\sigma^2}\bigg[\exp{\bigg(-\frac{(x_i+A\beta\cos{{\theta_0}})^2}{2\sigma^2}-\frac{(y_i+A\beta\sin{{\theta_0}})^2}{2\sigma^2} \bigg)} \nonumber\\
    &\qquad +\exp{\bigg(-\frac{(x_i-A\beta\cos{{\theta_0}})^2}{2\sigma^2}-\frac{(y_i-A\beta\sin{{\theta_0}})^2}{2\sigma^2} \bigg)}\nonumber\\
    &\qquad\qquad+\exp{\bigg(-\frac{(x_i+A\beta\cos{({\theta_0}+\Delta_1)})^2}{2\sigma^2}-\frac{(y_i+A\beta\sin{({\theta_0}+\Delta_1)})^2}{2\sigma^2} \bigg)}\nonumber\\
    &\qquad\qquad\qquad +\exp{\bigg(-\frac{(x_i-A\beta\cos{({\theta_0}+\Delta_1)})^2}{2\sigma^2}-\frac{(y_i-A\beta\sin{({\theta_0}+\Delta_1)})^2}{2\sigma^2} \bigg)} \bigg],\label{def:tildeQ'}\\
    \widehat{Q}_1^{\prime(n)}(x^n,y^n)&= \frac{1}{2}\prod_{i=1}^{n}\frac{1}{2}\frac{1}{{2\pi}\sigma^2}\bigg[\exp{\bigg(-\frac{(x_i+A\beta\cos{{\theta{'}}})^2}{2\sigma^2}-\frac{(y_i+A\beta\sin{{\theta{'}}})^2}{2\sigma^2} \bigg)}  \nonumber\\
    &\qquad+\exp{\bigg(-\frac{(x_i-A\beta\cos{{\theta{'}}})^2}{2\sigma^2}-\frac{(y_i-A\beta\sin{{\theta{'}}})^2}{2\sigma^2} \bigg)}\nonumber\\
    &\qquad\qquad+\exp{\bigg(-\frac{(x_i+A\beta\cos{({\theta{'}}+{\Delta_2})})^2}{2\sigma^2}-\frac{(y_i+A\beta\sin{({\theta{'}}+{\Delta_2})})^2}{2\sigma^2}\bigg)}\nonumber\\
    &\qquad\qquad\qquad +\exp{\bigg(-\frac{(x_i-A\beta\cos{({\theta{'}}+{\Delta_2})})^2}{2\sigma^2}-\frac{(y_i-A\beta\sin{({\theta{'}}+{\Delta_2})})^2}{2\sigma^2} \bigg)} \bigg].\label{def:hatQ'}
\end{align}\hrulefill
\end{figure*}

\section{Proof of arithmetic progression phase angles for the case that \texorpdfstring{$N=2$}{N=2}} 
\label{section:H}
In the case that $N=2$, we consider a general phase construction for the codebook 4-PSK. Specifically, the possible phase angle is either $\Delta_1$ or $\pi$. In other words, all possible codeword symbols are $(\beta,\Delta_1)$, $(-\beta,\Delta_1)$, $(\beta,\pi)$ and $(-\beta,\pi)$. Then, similar to the construction and calculation for the codebook 4-PSK with the arithmetic sequence phase construction, the probability distribution of $\mathbf{S_w}$ under $\mathcal{H}_1$ is given in \eqref{def:tildeQ'} on the top of the next page and the KL-divergence $\mathcal{D}(\widetilde{Q}_1^
{\prime n}\|{Q}^n_0)$ can be approximated as  
\begin{align}
&\mathcal{D}(\widetilde{Q}_1^{\prime n}\|{Q}^n_0)=\frac{nA^4\beta^4}{4\sigma^4}\Big(1-\frac{1}{2}\sin^2{\Delta_1}\Big)\nonumber\\
&\qquad\qquad\qquad-\frac{nA^6\beta^6}{6\sigma^6}\Big(1-\frac{1}{2}\sin^2{\Delta_1}\Big)+\mathcal{O}(\beta^8).
\end{align}

Clearly, when all possible phase angles are arithmetic sequence, i.e., $\Delta_1 =\pi/2$, the KL-divergence $\mathcal{D}(\widetilde{Q}_1^{\prime n}\|{Q}^n_0)$ achieves the minimum.


Further, we consider a general phase construction for the codebook 2-BPSK. The additional phase angle is selected uniformly at random from $\{\Delta_2,\pi\}$. Then, similar to the construction and calculation for the codebook 2-BPSK with the arithmetic sequence phase construction, the probability distribution of $\mathbf{S_w}$ under $\mathcal{H}_1$ is given in \eqref{def:hatQ'} on the top of the next page and the KL-divergence $\mathcal{D}(\widehat{Q}_1^{\prime(n)}\|{Q}^n_0)$ can be approximated as
\begin{align}
    &\mathcal{D}(\widehat{Q}_1^{\prime(n)}\|{Q}^n_0)\nonumber\\
    &=\frac{nA^4\beta^4}{8\sigma^4}\Big(1+2\cos^2{\Delta_2}-\cos^4{\Delta_2}\Big)+\mathcal{O}(\beta^6).
\end{align}

Clearly, when all possible phase angles are arithmetic sequence, i.e., $\Delta_2 =\pi/2$, the KL-divergence $\mathcal{D}(\widehat{Q}_1^{\prime(n)}\|{Q}^n_0)$ achieves the minimum.
\end{appendices}

\bibliography{IEEEabrv,reference}
\bibliographystyle{IEEEtran}

\end{document}